%% file: generative-demixing.tex
\documentclass{article}
\input{header.tex}

\usepackage{algorithm}

\title{Deep generative demixing: Recovering Lipschitz signals from noisy subgaussian mixtures}
\author[1]{Aaron Berk}

\affil[1]{%
  University of British Columbia %
  \protect\\ %
  Department of Mathematics %
  \protect\\ %
  1984 Mathematics Rd., Vancouver, BC, Canada V6T 1Z2%
}

\begin{document}
\maketitle

\begin{abstract}
  Generative neural networks (GNNs) have gained renown for efficaciously
  capturing intrinsic low-dimensional structure in natural images. Here, we
  investigate the subgaussian demixing problem for two Lipschitz signals, with
  GNN demixing as a special case. In demixing, one seeks identification of two
  signals given their sum and prior structural information. Here, we assume each
  signal lies in the range of a Lipschitz function, which includes many popular
  GNNs as a special case. We prove a sample complexity bound for nearly optimal
  recovery error that extends a recent result of Bora, \emph{et al.} (2017) from
  the compressed sensing setting with gaussian matrices to demixing with
  subgaussian ones. Under a linear signal model in which the signals lie in
  convex sets, McCoy \& Tropp (2014) have characterized the sample complexity
  for identification under subgaussian mixing. In the present setting, the
  signal structure need not be convex. For example, our result applies to a
  domain that is a non-convex union of convex cones. We support the efficacy of
  this demixing model with numerical simulations using trained GNNs, suggesting
  an algorithm that would be an interesting object of further theoretical study.
\end{abstract}

\section{Introduction}
\label{sec:introduction}

Generative neural networks (GNNs) are neural networks that are capable of
\emph{generating} data from a learned distribution. Recently, GNNs have gained
popularity for their performance on modelling spaces of natural
images~\cite{doersch2016tutorial, goodfellow2014generative,
  goodfellow2016tutorial, kingma2013auto, radford2015unsupervised}. For example,
variational autoencoders (VAEs) are composed of a ``decoder'' and ``encoder''
network, trained end-to-end to act as an approximate identity on the training
distribution. The VAE learns a low-dimensional representation for elements
within the data distribution~\cite{kingma2013auto} and its decoder is a type of
GNN.\@ For a thorough introduction to variational autoencoders, we
recommend~\cite{doersch2016tutorial}.

Recently, it was shown that some GNNs --- including decoders from popular VAE
architectures --- could be used as structural proxies for generalized
compressive sensing problems~\cite{bora2017compressed}. Roughly: it is possible
to approximately recover a signal from a noisy underdetermined linear system
with Gaussian measurements if the signal lies in the range of a Lipschitz
function. The work specialized its recovery result from arbitrary Lipschitz
functions, to the setting of Lipschitz-continuous neural networks. Many neural
networks commonly used in practice are Lipschitz, supporting the relevance of
this work to practical settings. For example,
see~\cite[Lemma~4.2]{bora2017compressed}.

Several extensions of this work have been
examined. In~\cite{wei2019statistical}, the authors derive an improved sample
complexity for compressed sensing-type problems when the underlying signal
structure is the range of a ReLU neural network. Another version of this problem
has been studied in which the measurements derive from a heavy-tailed, possibly
corrupted distribution~\cite{jalal2020robust}. A further generalization studies
the semi-parametrized single index model with possibly unknown possibly
nonlinear link function, and in which the measurement vectors derive from a
differentiable probability density function~\cite{wei2019statistical}. To our
knowledge, a theoretical analysis of the \emph{demixing problem} for two
Lipschitz functions remains open.

The demixing problem refers to the identification of two unknown signals from
their sum and given information about their structure.  Sharp recovery bounds
for \emph{convex demixing problems} have been thoroughly investigated
in~\cite{mccoy2014sharp}, where the source signals are each assumed to be well
encoded by some convex set. There, the authors show, under an \emph{incoherence}
condition, that recovery error depends on the \emph{effective dimension} of the
cones that govern the complexity of the source signals. Demixing with possibly
unknown, possibly nonlinear link functions has been
examined~\cite{soltani16demix, soltani2019provable}; however, the error bounds
derived using this form of analysis may be coarse and difficult to analyze in
practical settings~\cite{plan16_gener_lasso_with_non_linear_obser}.

The theory developed in~\cite{bora2017compressed} and~\cite{chen2018stable}
firmly lays the groundwork for the present work. Here, we investigate the
demixing problem for two signals $x^{*}, y^{*}$ that each lie in the range of an
arbitrary Lipschitz function. We assume that $x^{*}$ and $y^{*}$ are combined by
a \emph{subgaussian random mixing operation} that we describe below, and seek
approximate recovery of $(x^{*}, y^{*})$ \emph{with high probability} on the
mixing operation. The assumptions on the signal structure are mild, requiring
only that each signal lie in the range of a Lipschitz function. These Lipschitz
functions need not be the same, nor even have the same domain or range
dimensions.

To specify, let $G : \reals^{k} \to \reals^{n}$ and
$H : \reals^{k'}\to\reals^{m}$ be an $L_{G}$-Lipschitz and $L_{H}$-Lipschitz
function. We will be concerned with the setting in which $1 \leq k < n < \infty$
and $1 \leq k' < m < \infty$. We show that there are matrices
$A \in \reals^{m \times n}$ such that, if $x^{*} := G(u^{*}), y^{*} := H(v^{*})$
for points $(u^{*}, v^{*}) \in \reals^{k} \times \reals^{k'}$, and
\begin{align}
  \label{eq:demixing-model}
  b = A x^{*} + \sqrt m y^{*} + \eta
\end{align}
for some corruption $\eta \in \reals^{m}$, then $(x^{*}, y^{*})$ may be
approximated using only knowledge of $(b, A, G, H)$. Note the appearance of
$\sqrt m$ in~\eqref{eq:demixing-model} is for normalization purposes only, due
to the definition of $A$ below. The main result is stated
in~\autoref{thm:gen-demix-whp}.

A key tool in the proof of our results is a modified restricted isometry
property (RIP), first presented in~\cite{bora2017compressed}. There, the authors
proved the condition holds with high probability for matrices with normal random
entries. Here, we show that the condition holds with high probability for an
expanded class of matrices: namely, matrices $A$ with independent isotropic
subgaussian rows, as well as such concatenated matrices $[I\,A]$, where $I$ is
the identity matrix. Furthermore, we improve a result of~\cite{chen2018stable}
that establishes a deviation inequality for $[I\,A]$. We do this by improving
the dependence on the subgaussian constant for $A$ in the deviation inequality,
using an improved Bernstein's inequality first proved
in~\cite{jeong2019non}. This improvement is stated in~\autoref{thm:chen-1-new}.

Finally, we support our theoretical results with numerical simulations presented
in~\nameref{sec:numerics}. There, we use two trained VAEs to solve a noiseless
demixing problem. We visualize the recovery performance of our proposed
algorithm, and plot apparent exponential convergence of the objective function
\emph{vs}.\ iteration number. We include a discussion that highlights potential
topics of further interest for empirical investigations.

\paragraph{Motivation}

The demixing problem can be viewed both as a tool for signal transmission, and
as a tool for signal recovery. For example, the approach may be useful in the
lossy transmission of two signals over a channel, where the sender and receiver
possess the same Lipschitz continuous encoders $G$ and $H$, and where the
transmitted representation should use a small number of bits. In this scenario,
the mixture may be amenable to further compression or encoding steps --- such as
that with error-correcting codes, or other sparsifying
transforms~\cite{bora2017compressed,mackay2003information}.

In~\cite{chen2018stable}, the expression~\eqref{eq:demixing-model} is used to
describe the \emph{corrupted sensing} model in which a ground truth signal
$x^{*}$, after being encoded by $A$, is corrupted by (possibly unknown)
structured noise $y^{*}$, and unknown, unstructured noise $\eta$. This
interpretation translates to the present setting in the case where the
structured noise process giving rise to $y^{*}$ is well modelled by a Lipschitz
function $H$. Of interest to the present setting would be the case where either
$H$ is a Lipschitz deep neural network that has been trained to model some noise
process present in the corrupted measurements $b$.

\section{Background}
\label{sec:background}

We start by introducing the set-restriced eigenvalue condition that is critical
to the results developed in~\cite{bora2017compressed} and the present work. This
condition is the ``modified RIP'' hinted at above.

\begin{defn*}[{\cite[Definition 1]{bora2017compressed}}]
  \label{defn:bora-defn-1}
  Let $S \subseteq \reals^{n}$. For some parameters
  $\gamma > 0, \delta \geq 0$, a matrix $A \in \reals^{m \times n}$ is said to
  satisfy the $\SREC(S, \gamma, \delta)$ if $\forall x_{1}, x_{2} \in S$
  \begin{align*}
    \left\| A \left( x_{1} - x_{2} \right) \right\| %
    \geq \gamma \left\| x_{1} - x_{2} \right\| - \delta.
  \end{align*}
\end{defn*}

Next, we include for reference the relevant result of \cite{bora2017compressed},
which uses $\SREC$ for normal random matrices with iid entries to establish
recovery bounds for compressed sensing (CS) with signals in the range of
Lipschitz functions.

\begin{thm*}[{\cite[Theorem 1.2]{bora2017compressed}}]
  \label{thm:bora-1-2}
  Let $G: \mathbb{R}^{k} \rightarrow \mathbb{R}^{n}$ be an $L$-Lipschitz
  function. Let $\Phi \in \mathbb{R}^{m \times N}$ be a random Gaussian matrix
  for $m=\mathcal{O}\left(k \log \frac{L r}{\delta}\right)$, scaled so
  $\Phi_{ij} \iid \mathcal{N}(0,m^{-1})$. For any $x^{*} \in \reals^{N}$ and any
  observation $b = \Phi x^{*} + \eta$, let $\hat{u}$ minimize
  $\|b - \Phi G(\cdot)\|_{2}$ to within additive $\varepsilon$ of the optimum
  over vectors with $\|\hat{u}\|_{2} \leq r$. Then, with probability
  $1-\exp(-\Omega(m))$,
  \begin{align*}
    \left\|G(\hat{u})-x^{*}\right\|_{2} %
    \leq 6 \min _{u \in \mathbb{R}^{k} \atop\left\|u\right\|_{2} \leq r}
    \left\| G (u)-x^{*}\right\|_{2} + 3 \|\eta\|_{2} + 2 \varepsilon + 2 \delta.
  \end{align*}
\end{thm*}

In extending~\nameref{thm:bora-1-2} to the demixing setting
(\emph{cf}.~\autoref{thm:gen-demix-whp}), we allow for a broader class of random
matrices $A$ appearing in~\eqref{eq:demixing-model}. To define this class of
matrices, we start by defining a subgaussian random variable and subgaussian
random vector.

\begin{defn*}[Subgaussian]
  For a random variable, $X$, say $X$ is a \emph{subgaussian random variable}
  with subgaussian norm bounded by $K$, $\|X\|_{\psi_{2}} \leq K$, if for some
  $K > 0$,
  \begin{align*}
    \mathbb{P}\left(|X| \geq t \right) \leq 2 \exp\left(-t^{2} / K^{2}\right), %
    \quad \text{for all } t \geq 0. 
  \end{align*}
  For a random vector $X' \in \reals^{n}$, say $X$ is a \emph{subgaussian random
    vector} with $\|X'\|_{\psi_{2}} \leq K$ if
  \begin{align*}
    \sup_{e \in \sph^{n-1}}\|\ip{X, e}\|_{\psi_{2}} \leq K.
  \end{align*}
\end{defn*}

Additionally, we say that two random vectors $X, X' \in \reals^{n}$ are
\emph{isotropic} if $\E X' X^{T} = I_{n}$.

\begin{defn*}[$K$-subgaussian matrix]
  If $A \in \reals^{m \times n}$ is a matrix whose rows $A_{i} \in \reals^{n}$
  are independent, isotropic subgaussian random vectors with norm at most $K$,
  we call $A$ a $K$-subgaussian matrix.
\end{defn*}

Two geometric objects of particular importance to CS problems are the Gaussian
width and Gaussian complexity, which often appear in deviation inequalities used
to control $K$-subgaussian matrices on bounded sets.

\begin{defn*}[Gaussian width/complexity]
  Define the Gaussian width $\w(\mathcal{T})$ and Gaussian complexity
  $\gamma(\mathcal{T})$ of a set $\mathcal{T}\subseteq \reals^{n}$ by
  \begin{align*}
    \w(\mathcal{T}) &:= \E \sup_{x \in \mathcal{T}} \ip{x, g}, %
    & %
      \gamma(\mathcal{T}) &:= \E \sup_{x \in \mathcal{T}} |\ip{x, g}|, %
    & %
      g_{i} &\iid \mathcal{N}(0, 1). 
  \end{align*}
\end{defn*}

When $A$ is a $K$-subgaussian matrix, we define
$B = B(A) := \bmat{rr}{A & \sqrt m I}$ to be the concatenation of $A$ with a
scaled identity matrix (obtained by appending the columns of the second matrix
to the first). We may alternately refer to $B$ or $A$ as the ``subgaussian
mixing matrix'' or ``mixing matrix''; it will be clear from context whether we
mean $B$ or $A$. To establish recovery bounds for the Lipschitz demixing problem
with $K$-subgaussian matrices, we use an improved version of the following
deviation inequality for subgaussian mixing matrices $B$.

\begin{thm*}[{\cite[Theorem 1]{chen2018stable}}]
  \label{thm:chen-1}
  Let $A \in \reals^{m\times n}$ be a $K$-subgaussian matrix and
  $\mathcal{T} \subseteq \reals^{n} \times \reals^{m}$ be bounded. Then
  \begin{align*}
    \underset{(x, y) \in \mathcal{T}}{\mathbb{E}} \left| %
    \|A x + \sqrt{m} y \|_{2} - \sqrt{m} \cdot
    \sqrt{\|x\|_{2}^{2} + \|y\|_{2}^{2}}\right| %
    \leq CK^{2} \cdot \gamma (\mathcal{T})
  \end{align*}
  For any $t \geq 0$, the event
\begin{align*}
  \sup _{(x, y) \in \mathcal{T}} \left| \|A x + \sqrt{m} y\|_{2} %
  - \sqrt{m} \cdot \sqrt{\|x\|_{2}^{2}+\|y\|_{2}^{2}} \right| %
  \leq C K^{2}[\gamma(\mathcal{T})+t \cdot \rad(\mathcal{T})]
\end{align*}
holds with probability at least $1 - \exp \left(-t^{2}\right)$. 
\end{thm*}

Above, $\rad(\mathcal{T}):= \sup _{x \in \mathcal{T}}\|x\|_{2}$ denotes the
radius of the set $\mathcal{T}$. Our improvement to this result, given
in~\autoref{thm:chen-1-new}, is to establish an improved dependence on the
subgaussian constant $K$, from $K^{2}$ to $K\sqrt{\log{K}}$.

Finally, we require~\nameref{lem:bora-lem-4-3} to connect the aforementioned
pieces. It is worth noting that this lemma essentially requires no modification
from its original statement for it to be applicable to the demixing problem.

\begin{lem*}[{\cite[Lemma 4.3]{bora2017compressed}}]
  \label{lem:bora-lem-4-3}
  Let $A \in \mathbb{R}^{m \times n}$ be drawn from a distribution such that:
 \begin{enumerate}[itemsep=0pt, topsep=0pt]
 \item $\SREC(S, \gamma, \delta)$ holds with probability $1-p$;
   
 \item for every fixed $x \in \mathbb{R}^{n},\|A x\|_{2} \leq 2\|x\|_{2}$ with
   probability $1-p$.
 \end{enumerate}
 For any $x^{*} \in \mathbb{R}^{n}$ and noise $\eta \in \reals^{m}$, let
 $y = A x^{*} + \eta$. Let $\hat{x}$ satisfy
 \begin{align*}
   \| y - A \hat{x}\|_{2} \leq \min _{x \in S} \|y - A x\|_{2} + \epsilon.
 \end{align*}
 Then, with probability $1 - 2 p$,
 \begin{align*}
   \left\|\hat{x} - x^{*}\right\|_{2} %
   \leq \left( \tfrac{4}{\gamma}+1\right) \min _{x \in S} \left\| x^{*} - x \right\|_{2} %
   + \tfrac{1}{\gamma} ( 2\|\eta\|_{2} + \epsilon + \delta ).
 \end{align*}

\end{lem*}

As one final point regarding notation, we shall reserve the letters $c, C > 0$
for absolute constants, whose value may change from one appearance to the next.

\subsection{Two demixing modifications}
\label{sec:easy-problem}

Before moving on to the main results of this work, we discuss two modifications
of the demixing problem \eqref{eq:demixing-model}, which trivially reduce to an
application of \nameref{thm:bora-1-2}.

\subsubsection{One mixing matrix}
\label{sec:variant-1}

In the first modification, we assume that $(x^{*}, y^{*})$ are defined as
in~\eqref{eq:demixing-model}, and define instead
$b^{(1)} := \Phi(x^{*} + y^{*}) + \eta$, where
$\Phi_{ij} \iid \mathcal{N}(0, m^{-1})$ and $\eta \in \reals^{m}$ is some
unknown fixed corruption. In this setting, we show that the two signals may be
recovered but make no guarantee about correctly recovering each ground-truth
signal.

For simplicity, assume $k = k'$, $L_{G} = L_{H} = L$ and define
$w^{*} := (u^{*}, v^{*}) \in \reals^{2k}$. Similarly, write
$w := (u, v) \in \reals^{2k}$ and define the function
$F : \reals^{2k} \to \reals^{n}$ by $F(w) := G(u) + H(v)$ so that
$b^{(1)} := \Phi F(w^{*}) + \eta \in \reals^{m}$. Using \nameref{thm:bora-1-2},
one may determine some $\hat w \in \reals^{2k}$ that approximates $w^{*}$ with
high probability on the matrix $\Phi$.

\begin{prop}
  Suppose $m \geq C \cdot 2k\log\left(\frac{2L\|w^{*}\|}{\delta}\right)$.  Let
  $\hat w := (\hat u, \hat v) \in \reals^{2k}$ minimize $\|y - \Phi F(w)\|_{2}$
  within additive error $\varepsilon > 0$ over vectors $w \in \reals^{2k}$ with
  $\|w\| \leq \|w^{*}\|$. With probability at least $1 - \exp(- cm)$,
  \begin{align*}
    \|F(\hat w) - F(w^{*})\|_{2} %
    = \|G(\hat u) - G(u^{*}) + H(\hat v) - H(v^{*})\|_{2} %
    \leq 3 \|\eta\|_{2} + 2\varepsilon + 2\delta
  \end{align*}
\end{prop}

\begin{proof}
  The proof is a consequence of~\nameref{thm:bora-1-2}. Indeed, first notice
  that $F$ is $\sqrt 2L$ Lipschitz, because $G$ and $H$ are $L$-Lipschitz:
  \begin{align*}
    \|F(w) - F(w')\|_{2} %
    &\leq \|G(u) - G(u')\| + \|H(v) - H(v')\| %
      \leq L \|u - u'\| + L \|v - v'\| %
    \\
    &\leq \sqrt 2L \sqrt{\|u - u'\|^{2} + \|v - v'\|^{2}} %
      = \sqrt 2L \|w - w'\|_{2}. 
  \end{align*}
  Next observe that $F$ is a $d$-layer ReLU network because $G$ and $H$
  are. Indeed, the linear transformations
  $W_{i}^{G} \in \reals^{k_{i}\times k_{i-1}}, W_{i}^{H} \in \reals^{k_{i}'
    \times k_{i-1}'}$ used in defining $G$ and $H$ can, at each layer, be merged
  into block-diagonal transformations
  \begin{align*}
    W_{i} := \left[
    \begin{array}{cc}
      W_{i}^{G} & \mathbf{0} \\
      \mathbf{0} & W_{i}^{H}
    \end{array} \right] %
    \in \reals^{(k_{i} + k_{i}') \times (k_{i-1} + k_{i-1}')}, %
    i \in 1,\ldots, d.
  \end{align*}
  to obtain
  \begin{align*}
    F(w) = \left[\mathrm{Id}_{n}\, \mathrm{Id}_{n}\right] \circ \sigma\circ W_{d}\circ
    \sigma\circ W_{d-1}\circ \cdots \circ \sigma \circ W_{1} (w).
  \end{align*}
  Applying~\nameref{thm:bora-1-2} thereby completes the proof.
\end{proof}

\subsection{Two mixing matrices}
\label{sec:variant-2}

Alternately, one may assume that $x^{*}$ and $y^{*}$ are mixed as
$b^{(2)} := \Phi_{1} x^{*} + \Phi_{2}y^{*} + \eta$, where $\Phi_{i}, i=1,2$ are
independent and have the same distribution as $\Phi$ above. Indeed, observe that
$b^{(2)}$ may be written as
\begin{align*}
  b^{(2)} &:= \overline{\Phi} \bmat{r}{G(u^{*})\\ H(v^{*})} + \eta %
          & %
  \overline{\Phi} &:= \bmat{rr}{\Phi_{1} & \Phi_{2}}. %
\end{align*}
Since the concatentated matrix $\overline{\Phi}$ has entries
$\overline{\Phi}_{ij} \iid \mathcal{N}(0, m^{-1})$, a similar line of reasoning
as in \nameref{sec:variant-1} using $\overline{F}(w) := \bmat{r}{G(u)\\H(v)}$
allows the approximate recovery $\bmat{r}{\hat x \\ \hat y} \in \reals^{2n}$ of
$\bmat{r}{x^{*}\\y^{*}} \in \reals^{2n}$ along with codes
$\hat u, \hat v \in \reals^{k}$ with $\hat x = G(\hat u)$ and
$\hat y = H(\hat v)$.

We omit a formal proposition of this fact, noting that it would essentially be a
re-wording of~\nameref{thm:bora-1-2}.




\section{Subgaussian Lipschitz demixing}
\label{sec:gen-demix-random-mapping}

Throughout this section, we assume the model~\eqref{eq:demixing-model}, with
$A \in \reals^{m\times n}$ being a $K$-subgaussian matrix. We further assume
$u^{*} \in B^{k}(r) \subseteq \reals^{k}, v^{*} \in B^{k'}(r') \subseteq
\reals^{k'}$ for some $r, r' > 0$, with $x^{*} := G(u^{*})$ and
$y^{*} := H(v^{*})$. In particular, the model~\eqref{eq:demixing-model} becomes
$b^{*} := AG(u^{*}) + H(v^{*}) + \eta$ where $\eta \in \reals^{m}$. Roughly,
this section proves that if $m$ is sufficiently large, then with high
probability on the realization of $A$, $\hat x = G(\hat u)$ and
$\hat y = H(\hat v)$ may be obtained from $(b, A, G, H)$ with
$\hat x \approx x^{*}$ and $\hat y \approx y^{*}$.

For notational brevity, we write $\mathbb{B} := B^{k}(r) \times B^{k'}(r')$ and
further write $F(\mathbb{B}) := G(B^{k}(r)) \times H(B^{k'}(r'))$. Finally,
throughout this work we define $\tilde{K} := K \sqrt{\log K}$. As per the remark
at the beginning of~\cite[$\S\,4.3$]{jeong2019non}, if $X$ is a subgaussian
random variable and $\E |X|^{2} = 1, \|X\|_{\psi_{2}} \leq K$, then
$K \geq K_{0} := \left( \log 2 \right)^{-1/2} \approx 1.201$. In particular,
mandating isotropy of the rows implies $1 < K \leq \tilde K \leq K^{2}$.

The first lemma is a modification of~\cite[Lemma 8.2]{bora2017compressed},
adapted for the demixing problem. It establishes that mixing signals from
$F(\mathbb{B})$ is well approximated by mixing signals from a net for
$F(\mathbb{B})$.

\begin{lem}
  \label{lem:berk-8-2}
  Fix $\delta > 0$, define $\kappa := m / n$ and suppose
  \begin{align*}
    m \geq Ck \log\left(\frac{crL_{G}\tilde K}{\delta}\right) + Ck'
    \log\left(\frac{cr'L_{H}\tilde K}{\delta}\right).
  \end{align*}
  There is a net $M \subseteq \reals^{k} \times \reals^{k'}$ for
  $\mathbb{B} \subseteq \reals^{k} \times \reals^{k'}$ such that, for any
  $z \in F(\mathbb{B})$, if
  \begin{align*}
    z' \in \argmin_{\hat z \in F(M)} \|z - \hat z\|_{2}
  \end{align*}
  then $m^{-1/2}\|B(z - z')\| \leq C \delta$ with probability at least
  ${1 - \exp(-cm)}$. Moreover, where $C_{\kappa} > 0$ is an absolute constant
  depending only on $\kappa$, the cardinality of $M$ satisfies
  \begin{align*}
    \log |M| %
    \leq C_{\kappa}(k + k') %
    + Ck \log\left(\frac{crL_{G}\tilde K}{\delta}\right)
    + Ck' \log\left(\frac{cr'L_{H}\tilde K}{\delta}\right).
  \end{align*}

\end{lem}

Next, we establish that normalized subgaussian mixing matrices
satisfy~\nameref{defn:bora-defn-1} with high probability.

\begin{lem}
  \label{lem:s-rec}
  For $\alpha < 1$, if
  \begin{align*}
    m \geq \frac{C \tilde{K}^{2}}{\alpha^{2}}  \left( %
    k\log\left(\frac{L_{G}r}{\delta}\right) %
    + k' \log \left(\frac{L_{H}r'}{\delta}\right) %
    \right),
  \end{align*}
  then $m^{-1/2}B$ satisfies
  $\operatorname{S-REC}(F(\mathbb{B}), 1 - \alpha, \delta)$ with probability at
  least $1 - \exp(-cm)$.
\end{lem}

The main result of this work follows by combining~\autoref{lem:berk-8-2}
and~\ref{lem:s-rec} with~\nameref{lem:bora-lem-4-3}.

\begin{thm}[Lipschitz demixing]
  \label{thm:gen-demix-whp}
  Fix $\delta > 0$ and let $A \in \reals^{m \times n}$ be a $K$-subgaussian
  matrix. Suppose $G : \reals^{k} \to \reals^{n}$ is $L_{G}$-Lipschitz,
  $H : \reals^{k'}\to\reals^{m}$ is $L_{H}$-Lipschitz. Let
  $z^{*} = [G(u^{*})^{T} H(v^{*})^{T}]^{T}$ for some
  $(u^{*}, v^{*}) \in \mathbb{B}$. Suppose $b = Bz^{*} + \eta$ for some
  $\eta \in \reals^{m}$. Suppose
  $\hat z \in F(\mathbb{B}) := G(B^{k}(r))\times H(B^{k'}(r'))$ satisfies
  $\|B\hat z - b \|_{2} \leq \min_{z \in F(\mathbb{B})} \|Bz - b\|_{2} +
  \varepsilon$. If
  \begin{align*}
    m %
    \geq C \tilde{K}^{2} \big( %
    k \log\Big(\frac{L_{G}r \tilde K}{\delta}\Big) %
    + k' \log \Big(\frac{L_{H}r' \tilde K}{\delta}\Big) %
    \Big),
  \end{align*}
  then with probability at least $1 - \exp(- cm)$ on the realization of $A$, it
  holds that
  \begin{align*}
    \|\hat z - z^{*}\| %
    \leq C \tilde K \cdot \min_{z \in F(\mathbb{B})}\|z - z^{*}\| %
    + C \left(2 \|\eta\| + \delta + \varepsilon\right). 
  \end{align*}
\end{thm}

We defer the proofs of~\autoref{lem:berk-8-2},~\ref{lem:s-rec}
and~\autoref{thm:gen-demix-whp}
to~\nameref{sec:proofs-gen-demix-random-mapping}.

\section{Improved subgaussian constant}
\label{sec:impr-k-depend}

Let $A \in \reals^{m \times n}$ be a $K$-subgaussian matrix. Where
$z = \bmat{r}{x\\ y} \in \reals^{n}\times\reals^{m}$, define the centered random
process $X_{z} = X_{x,y}$ by 
\begin{align}
  \label{eq:def-Xz}
  X_{z} %
  &:= \|Ax + \sqrt m y\|_{2} - \left(\E \|A x + \sqrt m y \|_{2}^{2}\right)^{1/2} %
    \nonumber
  \\
  & = \|Ax + \sqrt m y\|_{2} - \sqrt m \cdot \sqrt{\|x\|_{2}^{2} + \|y\|_{2}^{2}}
    \nonumber
  \\
  & = \|Bz\|_{2} - \sqrt m \|z\|_{2}
\end{align}
That $X_{z}$ is subgaussian (indeed, a process with subgaussian increments) was
established by~\nameref{thm:chen-1}. In this section we improve the dependence
of $X_{z}$ on $K$ from $K^{2}$ to $K\sqrt{\log K}$. The tools required for this
improvement were developed in~\cite{jeong2019non}.

\begin{lem}[Improved subgaussian increments]
  \label{lem:subgaussian-process-better-constant}
  Let $A \in \reals^{m \times n}$ be a $K$-subgaussian matrix. The centered
  process $X_{z}$ has subgaussian increments:
  \begin{align*}
    \|X_{z} - X_{z'}\|_{\psi_{2}} %
    &\leq C K \sqrt{\log K} \|z-z'\|_{2}, %
    &
      &\text{for every } z, z' \in \reals^{n}\times \reals^{m}.
  \end{align*}
  
\end{lem}

The proof of this result is technical, but strongly resembles the one presented
in~\cite{chen2018stable}. To
prove~\autoref{lem:subgaussian-process-better-constant} essentially requires
replacing each application of Bernstein's inequality in the proof
of~\cite[Lemma~5]{chen2018stable} with one
of~\cite[New~Bernstein's~inequality]{jeong2019non}. We provide a detailed proof
of~\autoref{lem:subgaussian-process-better-constant}
in~\nameref{sec:impr-subg-const}.

Having established this result, it is straightforward to
combine~\autoref{lem:subgaussian-process-better-constant} with Talagrand's
majorizing measure theorem~\cite[Theorem~8.6.1]{vershynin2018high} to obtain the
desired improvement of~\nameref{thm:chen-1}. We relegate these technical details
to~\nameref{sec:impr-subg-const}.

\begin{thm}[Improved deviation inequality]
  \label{thm:chen-1-new}
  Let $A \in \reals^{m\times n}$ be a $K$-subgaussian matrix and
  $\mathcal{T} \subseteq \reals^{n} \times \reals^{m}$ be bounded. Then,
  \begin{align*}
    \underset{(x, y) \in \mathcal{T}}{\mathbb{E}} \left| %
    \|A x + \sqrt{m} y \|_{2} - \sqrt{m} \cdot
    \sqrt{\|x\|_{2}^{2} + \|y\|_{2}^{2}}\right| %
    \leq C\tilde K \cdot \gamma (\mathcal{T}).
  \end{align*}
  Moreover, for any $t \geq 0$, the event
  \begin{align*}
    \sup _{(x, y) \in \mathcal{T}} \left| \|A x + \sqrt{m} y\|_{2} %
    - \sqrt{m} \cdot \sqrt{\|x\|_{2}^{2}+\|y\|_{2}^{2}} \right| %
    \leq C \tilde K [\gamma(\mathcal{T})+t \cdot \rad(\mathcal{T})]
  \end{align*}
  holds with probability at least $1 - \exp \left(-t^{2}\right)$.
\end{thm}

\input{numerics}

\input{conclusion}

\section*{Acknowledgments}
\label{sec:acknowledgments}

I would like to thank Drs. Babhru Joshi, Yaniv Plan, and \"Ozg\"ur Yilmaz for
helpful discussion in the development and communication of this work.

\bibliographystyle{plain}
\bibliography{/Users/aberk/Dropbox/org/bibliography}

\appendix

\input{proofs-gen-demix-random-mapping}

\input{proofs-improved-subgaussian-constant}

\input{supplementary-numerics}

\end{document}

%% file: header.tex
\usepackage[bb=boondox]{mathalfa}
\usepackage{comment}
\usepackage{Tabbing}
\usepackage[dvipsnames,table]{xcolor}
\usepackage{url}
\usepackage{fancyhdr}
\usepackage[top= 1.3 in, bottom =1 in, left =0.75 in, right = 1 in]{geometry}
\usepackage{amssymb}
\usepackage{amsmath}
\usepackage{amsthm}
\usepackage{mathrsfs}
\usepackage{graphicx}
\usepackage{enumitem}
\usepackage[toc, page]{appendix}
\usepackage{mathtools}
\usepackage{caption}
\usepackage{subcaption}
\usepackage{listings}
\usepackage[colorlinks = true, linkcolor = blue, urlcolor=blue, citecolor = blue]{hyperref}
\usepackage{multicol}
\usepackage{sidecap}
\usepackage{tikz-cd}
\usepackage{booktabs}

\usepackage[OT2,T1]{fontenc}
\DeclareSymbolFont{cyrletters}{OT2}{wncyr}{m}{n}
\DeclareMathSymbol{\Sha}{\mathalpha}{cyrletters}{"58}

\usepackage{aliascnt}

\usepackage{authblk}

\makeatletter

\renewcommand{\d}{\ensuremath{\,\mathrm{d}}}
\let\oldequation\equation
\renewcommand{\equation}{\@hyper@itemfalse\oldequation}
\let\oldgather\gather
\renewcommand{\gather}{\@hyper@itemfalse\oldgather}
\let\oldalign\align
\renewcommand{\align}{\@hyper@itemfalse\oldalign}
\makeatother

\definecolor{codegray}{gray}{0.95}
\definecolor{CODEGRAY}{gray}{0.95}

\usepackage{soul}
\sethlcolor{codegray}
\newcommand{\inlinecode}[1]{\hl{\tt #1}}

\newcommand{\ie}{\emph{i.\@e.\@,~}}
\newcommand{\eg}{\emph{e.\@g.\@,~}}

\newcommand{\reals}{\ensuremath{\mathbb{R}}}

\newcommand{\sph}{\ensuremath{\mathbb{S}}}
 
\newcommand{\ip}[1]{\ensuremath{\langle #1\rangle}}

\newcommand{\iid}{\ensuremath{\overset{\text{iid}}{\sim}}}

\newcommand{\bmat}[2]{\ensuremath{\left[
      \begin{array}{#1}
        #2
      \end{array}
    \right]}}

\theoremstyle{plain}
\newtheorem{thm}{Theorem}
\newtheorem*{thm*}{Theorem}
\newtheorem*{lem*}{Lemma}

\newaliascnt{prop}{thm}
\newtheorem{prop}[prop]{Proposition}
\aliascntresetthe{prop}

\newaliascnt{coro}{thm}

\aliascntresetthe{coro}

\newaliascnt{lem}{thm}
\newtheorem{lem}[lem]{Lemma}
\aliascntresetthe{lem}

\theoremstyle{definition}

\newtheorem*{defn*}{Definition}

\theoremstyle{remark}
\newtheorem{fact}{Fact}
\newtheorem{rmk}[fact]{Remark}
\newtheorem*{rmk*}{Remark}
\newtheorem*{fact*}{Fact}

\DeclarePairedDelimiterX{\kldivx}[2]{(}{)}{%
  #1\;\delimsize\|\;#2%
}
\newcommand{\kldiv}{D_{KL}\kldivx}

\DeclareMathOperator{\rad}{rad}
\DeclareMathOperator{\diam}{diam}
\DeclareMathOperator{\SREC}{S-REC}

\DeclareMathOperator*{\argmin}{\arg\min}
\DeclareMathOperator*{\E}{\mathbb E}

\DeclareMathOperator{\w}{w}

\usepackage{xspace}
 
\setlist[enumerate, 1]{align=left, label=\arabic*.}
\setlist[enumerate, 2]{align=left, label=(\alph*)}
\setlist[itemize,1]{align=left, topsep=-2pt}


%% file: numerics.tex
\section{Numerics}
\label{sec:numerics}

In this section we describe numerical results supporting the theory developed
in~\nameref{sec:gen-demix-random-mapping}. We include only the main results of
this simulation; specific implementation details beyond the general experimental
set-up and simulation results may be found in~\nameref{sec:s1-numerics}.


For this simulation, two deep convolutional VAEs were trained on images of $1s$
and $8$s, respectively; the trained decoders
$D_{i}: \reals^{128} \to [0,1]^{28\times 28}$, $i = 1, 8$ from each VAE were
used as the Lipschitz-continuous deep generative networks for the demixing
problem. A description of the network architecture and training method for the
VAEs is detailed in~\nameref{sec:s1-numerics}. We claim no novelty for our VAE
or decoder architecture; extensions of VAEs involving convolutional layers
appear in existing work~\cite{kulkarni2015deep,pu2016variational}. The digits
used to train each VAE were a subclass of the MNIST
database~\cite{deng2012mnist} containing only images of the digits $1$ and
$8$. Each group of digits was randomly parititioned into training, validation
and test sets. Only training and validation sets were used in the training of
the VAEs.

Two ground-truth images (\emph{cf.} \autoref{fig:mnist-vae-true0} \&
\ref{fig:mnist-vae-true1}, resp.) were randomly sampled from the test partition
of each dataset and mixed according to~\eqref{eq:demixing-model}, where the
mixing matrix had iid normal random entries
(\emph{cf.}~\autoref{fig:mnist-vae-mixture}). In particular, where
$x_{8} \in [0,1]^{28\times 28}$ is the (vectorized) grayscale image of an $8$
and $x_{1} \in [0,1]^{28 \times 28}$ that of a $1$, the mixture is given by
\begin{align}
  \label{eq:b-mnist}
  b &= \Phi x_{8} + x_{1}, %
    &
  \Phi &\in \reals^{784 \times 784}, %
  \Phi_{ij} \iid \mathcal{N}(0, \tfrac{1}{784}).
\end{align}
We did not investigate the effect of additive noise in this simulation, and
leave that to a future work.

To solve the demixing problem, we used a variant of minibatch stochastic
gradient descent~\cite{kingma2014adam} to minimize the mean-squared error (MSE)
loss between the predicted mixture and the true mixture (using
PyTorch~\cite{paszke2017automatic}). Details for this numerical implementation,
including pseudo-code, are provided in~\nameref{sec:s1-numerics}. The
recovered mixture is depicted in~\autoref{fig:mnist-vae-mixture_pred}, and the
recovered signals, lying in the output of each generator $D_{i}$, are shown
in~\autoref{fig:mnist-vae-recovered0} and~\ref{fig:mnist-vae-recovered1} for
$i = 1, 8$, respectively.

Convergence of the numerical implementation of the demixing algorithm is shown
in~\autoref{fig:mnist-vae-demix-loss-plot}, where we plot the MSE loss between
the true mixture and the predicted mixture as a function of the iteration
number. The recovered image approximating $x_{1}$ had an MSE of
$1.59\cdot 10^{-3}$; that for $x_{8}$, $7.48 \cdot 10^{-4}$; and that for the
mixture, $3.92\cdot 10^{-4}$.

\begin{figure*}[t]
  \centering
  \begin{subfigure}[h]{0.15\linewidth}
    \centering
    \includegraphics[width=\linewidth]{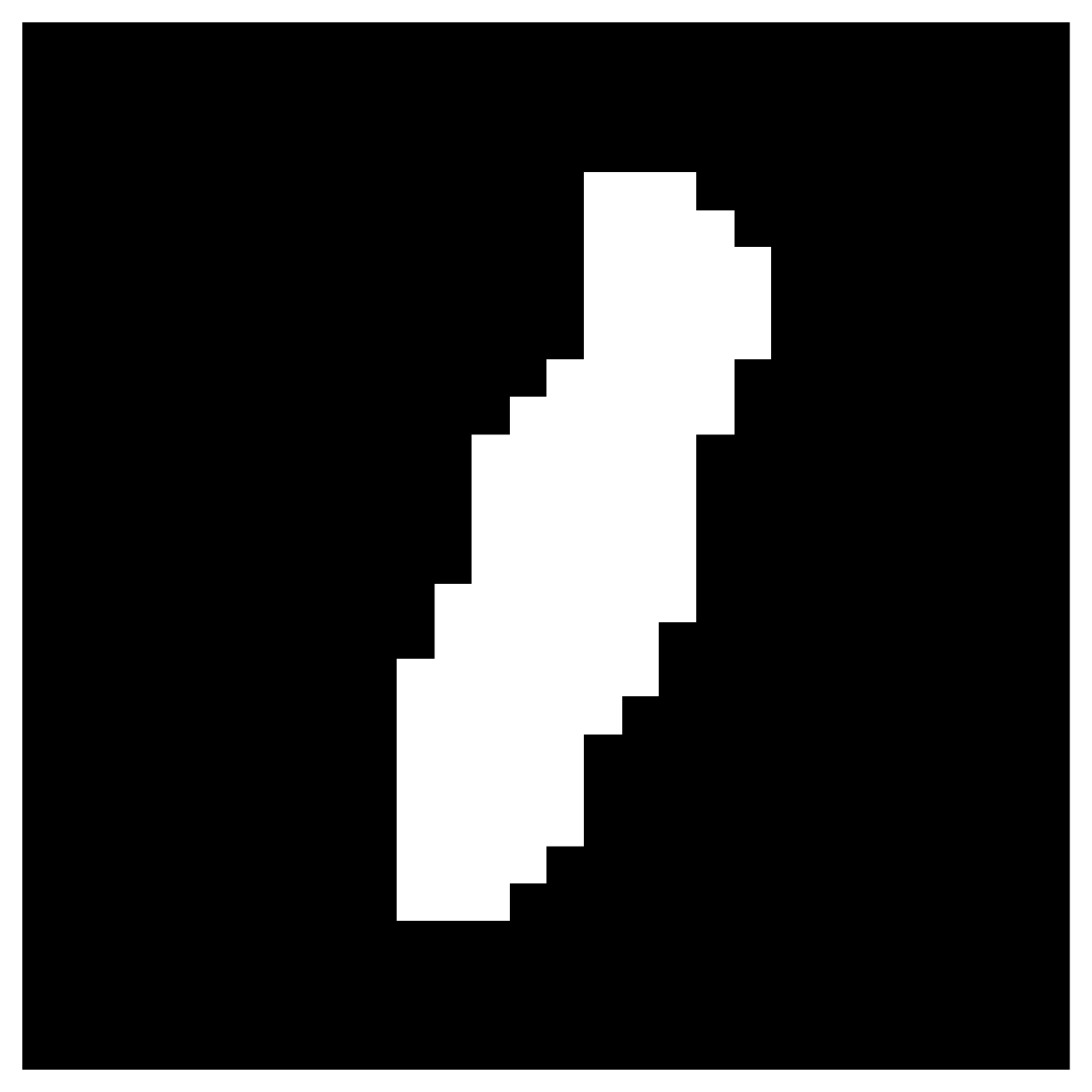}
    \caption{$x_{1}$}
    \label{fig:mnist-vae-true0}
  \end{subfigure}\hfill
  \begin{subfigure}[h]{0.15\linewidth}
    \centering
    \includegraphics[width=\linewidth]{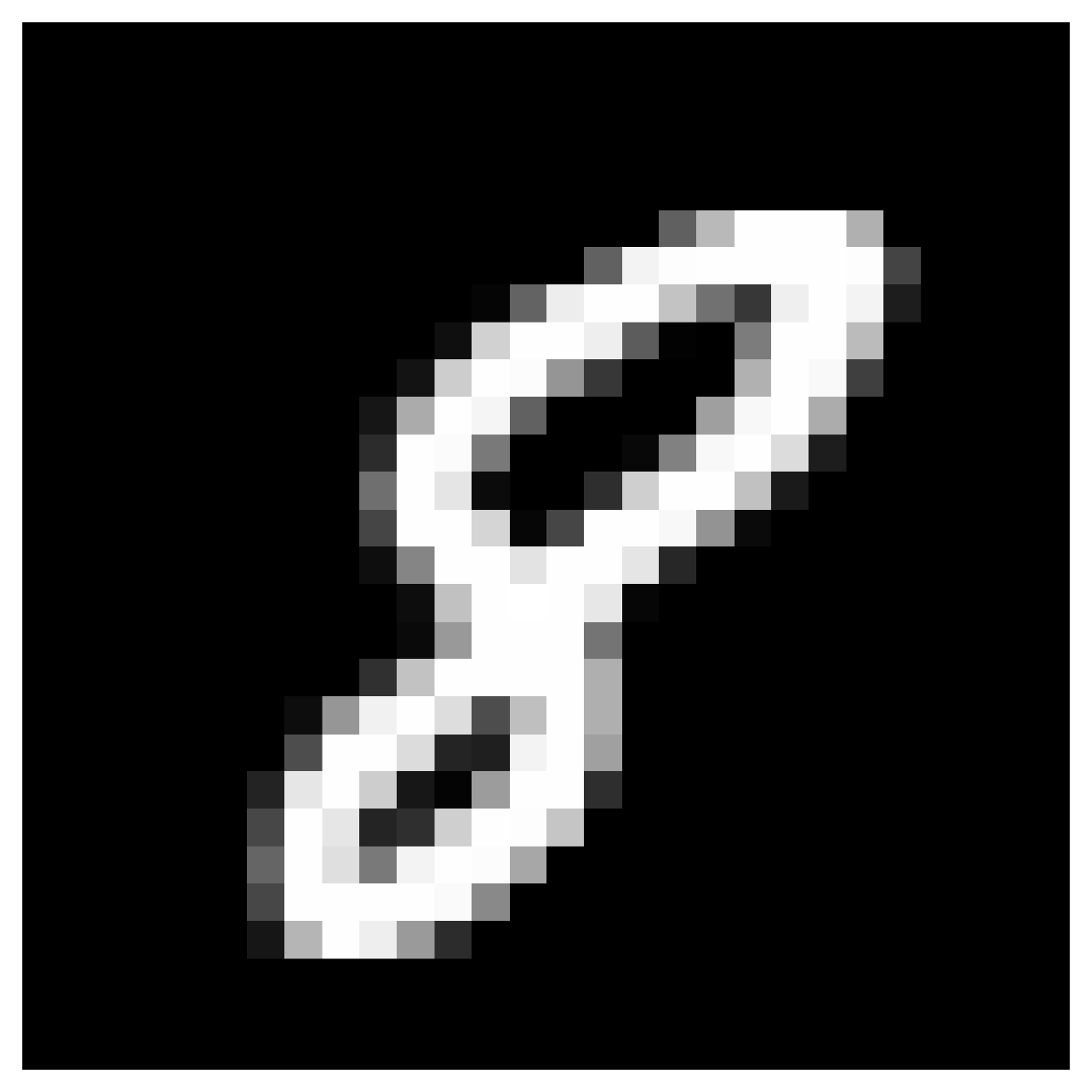}
    \caption{$x_{8}$}
    \label{fig:mnist-vae-true1}
  \end{subfigure}\hfill
  \begin{subfigure}[h]{0.15\linewidth}
    \centering
    \includegraphics[width=\linewidth]{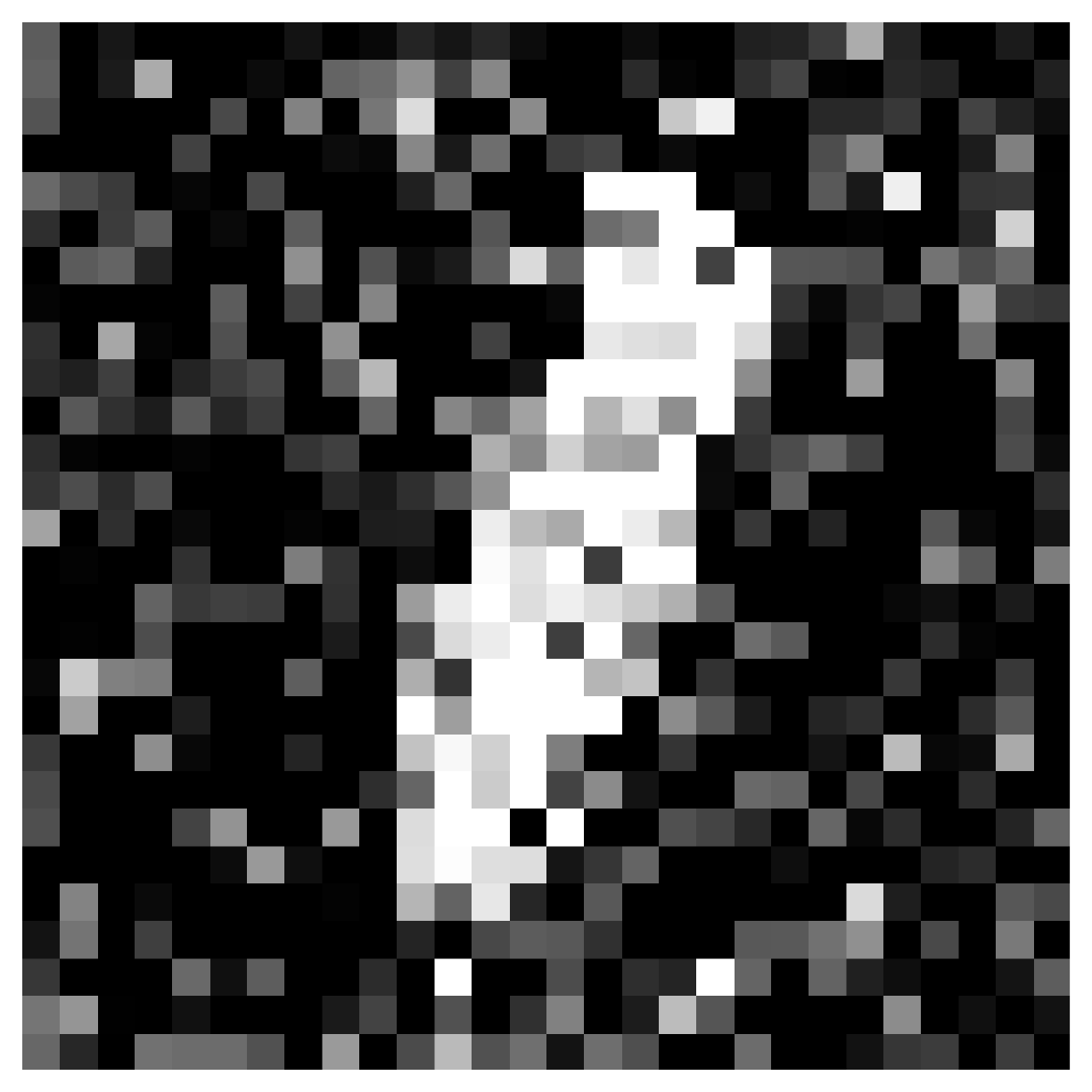}
    \caption{$b$}
    \label{fig:mnist-vae-mixture}
  \end{subfigure}\hfill
  \begin{subtable}[h]{0.35\linewidth}
    \centering
    \begin{tabular}[b]{cr}
      \toprule
      image & MSE\\
      \midrule
      $1$ & $1.59\cdot 10^{-3}$
      \\
      $8$ & $7.48 \cdot 10^{-4}$
      \\
      mixture & $3.92\cdot 10^{-4}$
      \\
      \bottomrule
    \end{tabular}
    \caption{Final MSE}
    \label{tab:mse}
  \end{subtable}

  \begin{subfigure}[h]{0.15\linewidth}
    \centering
    \includegraphics[width=\linewidth]{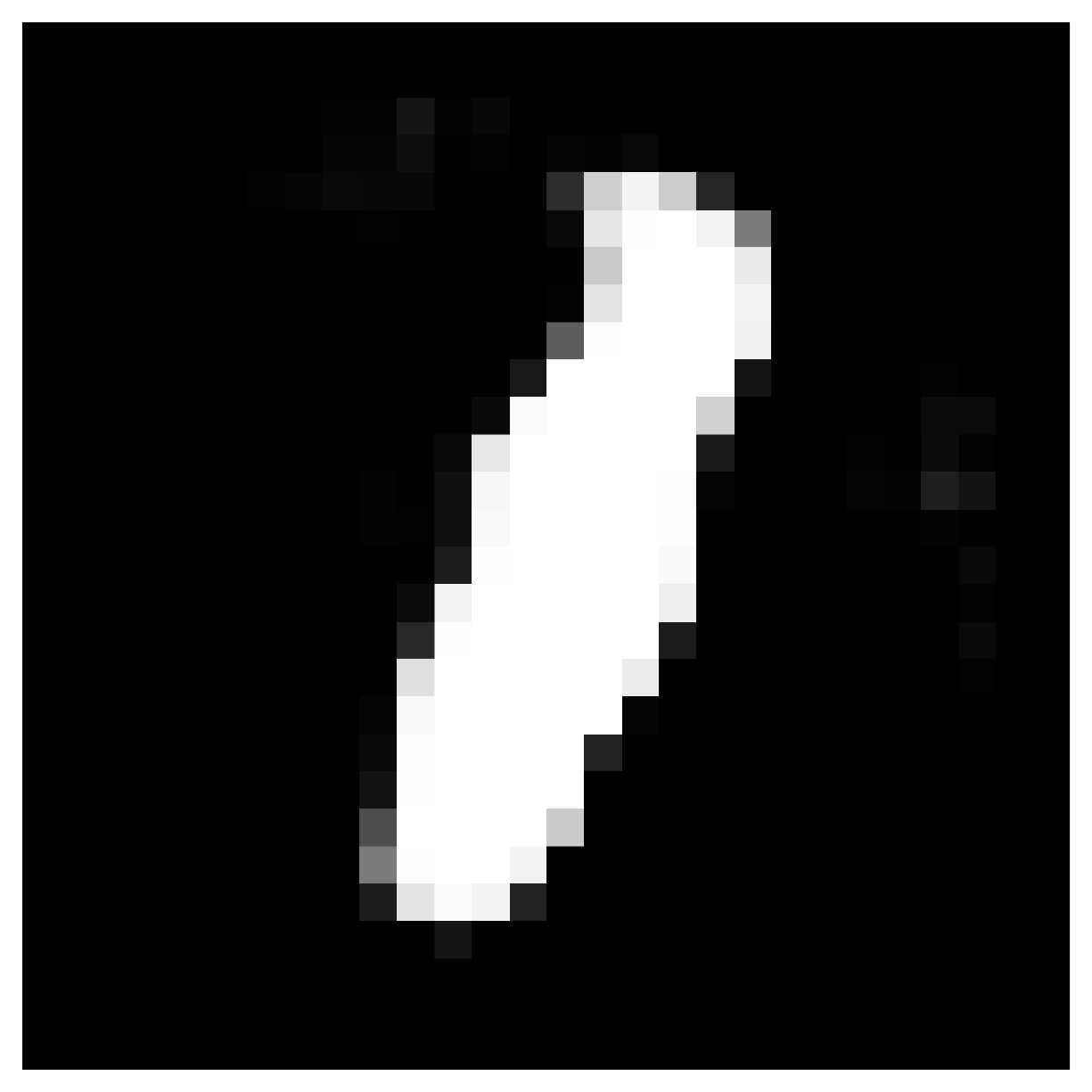}
    \caption{$\hat x_{1}$}
    \label{fig:mnist-vae-recovered0}
  \end{subfigure}\hfill
  \begin{subfigure}[h]{0.15\linewidth}
    \centering
    \includegraphics[width=\linewidth]{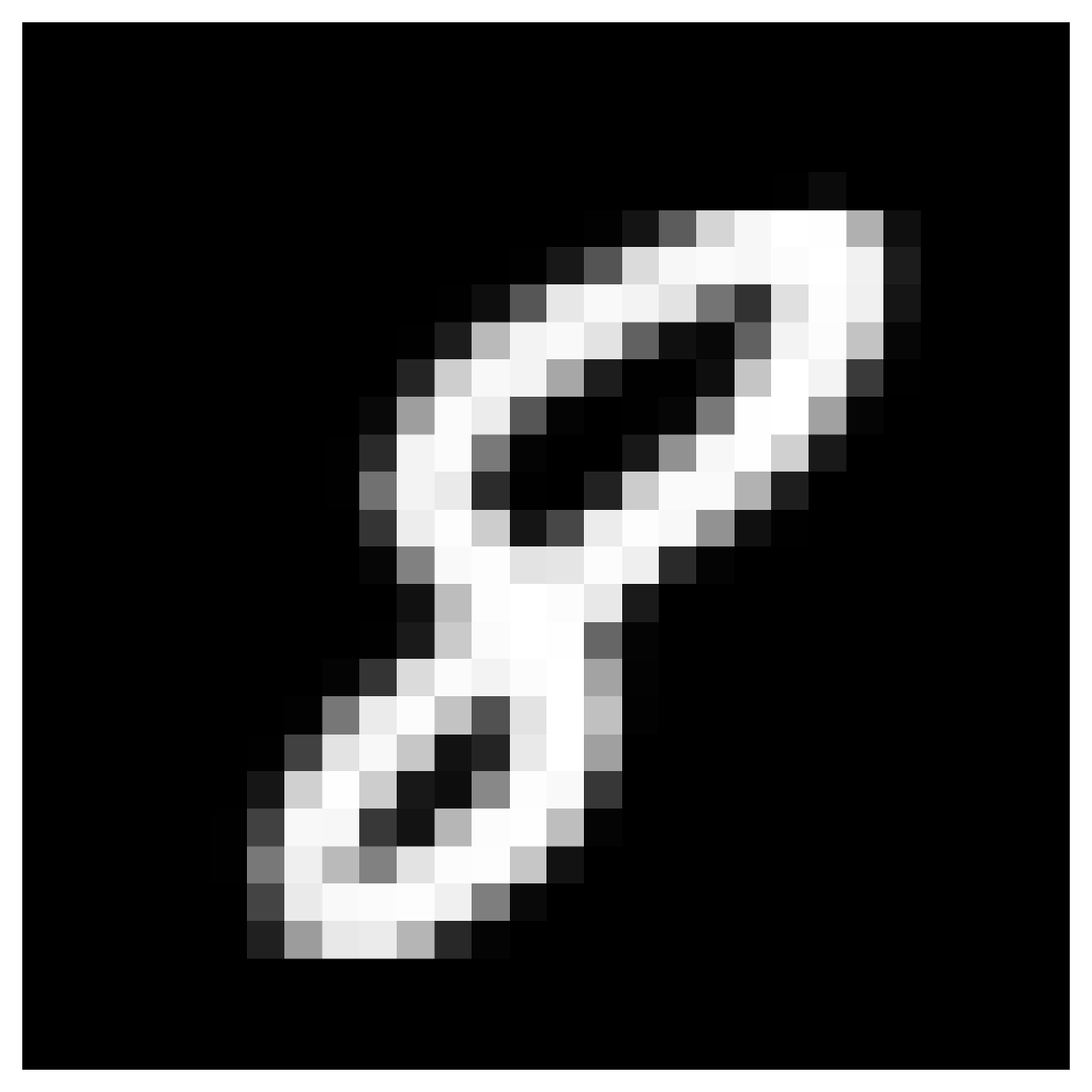}
    \caption{$\hat x_{8}$}
    \label{fig:mnist-vae-recovered1}
  \end{subfigure}\hfill
  \begin{subfigure}[h]{0.15\linewidth}
    \centering
    \includegraphics[width=\linewidth]{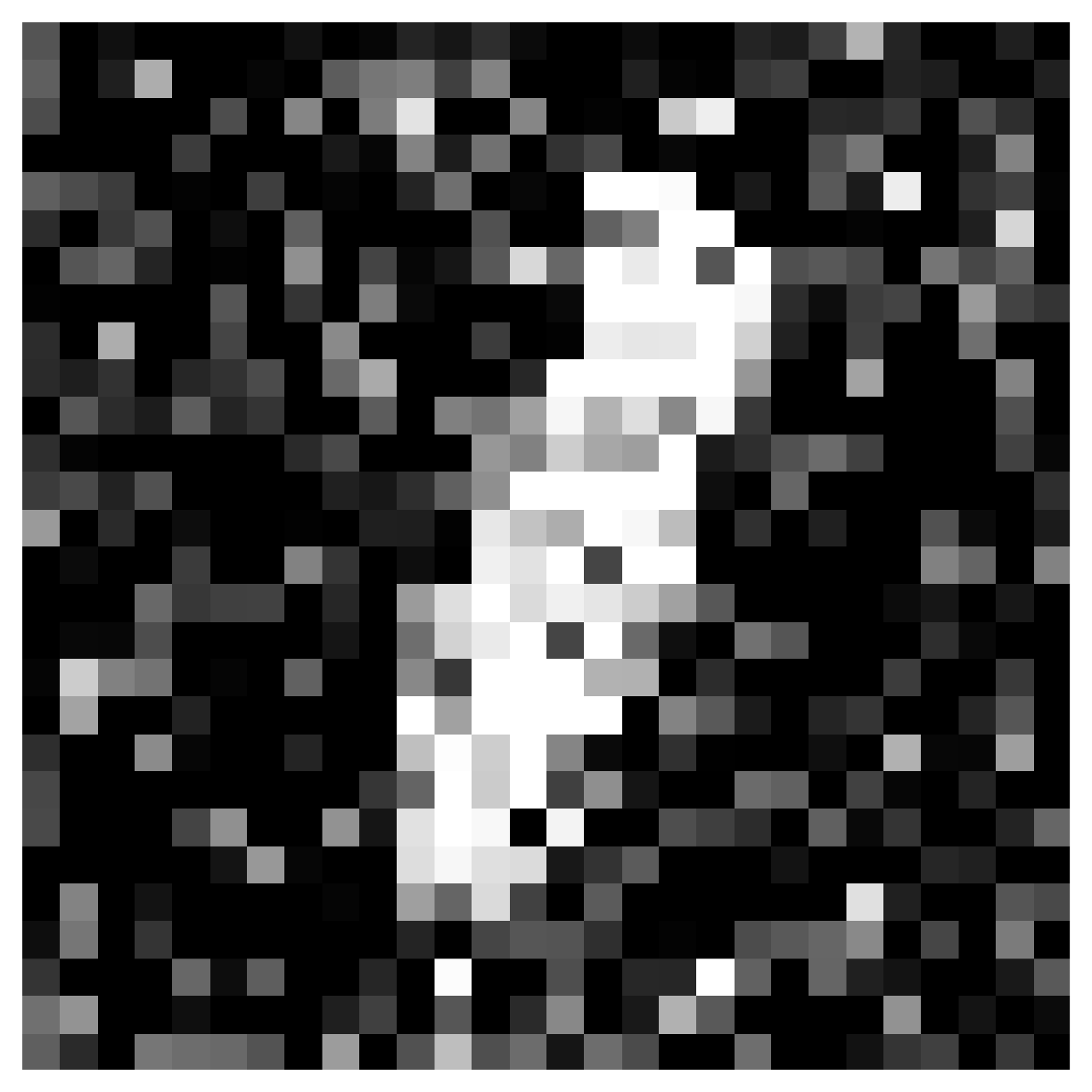}
    \caption{$\hat b$}
    \label{fig:mnist-vae-mixture_pred}
  \end{subfigure}\hfill
  \begin{subfigure}[h]{.35\linewidth}
    \centering
    \includegraphics[width=\linewidth]{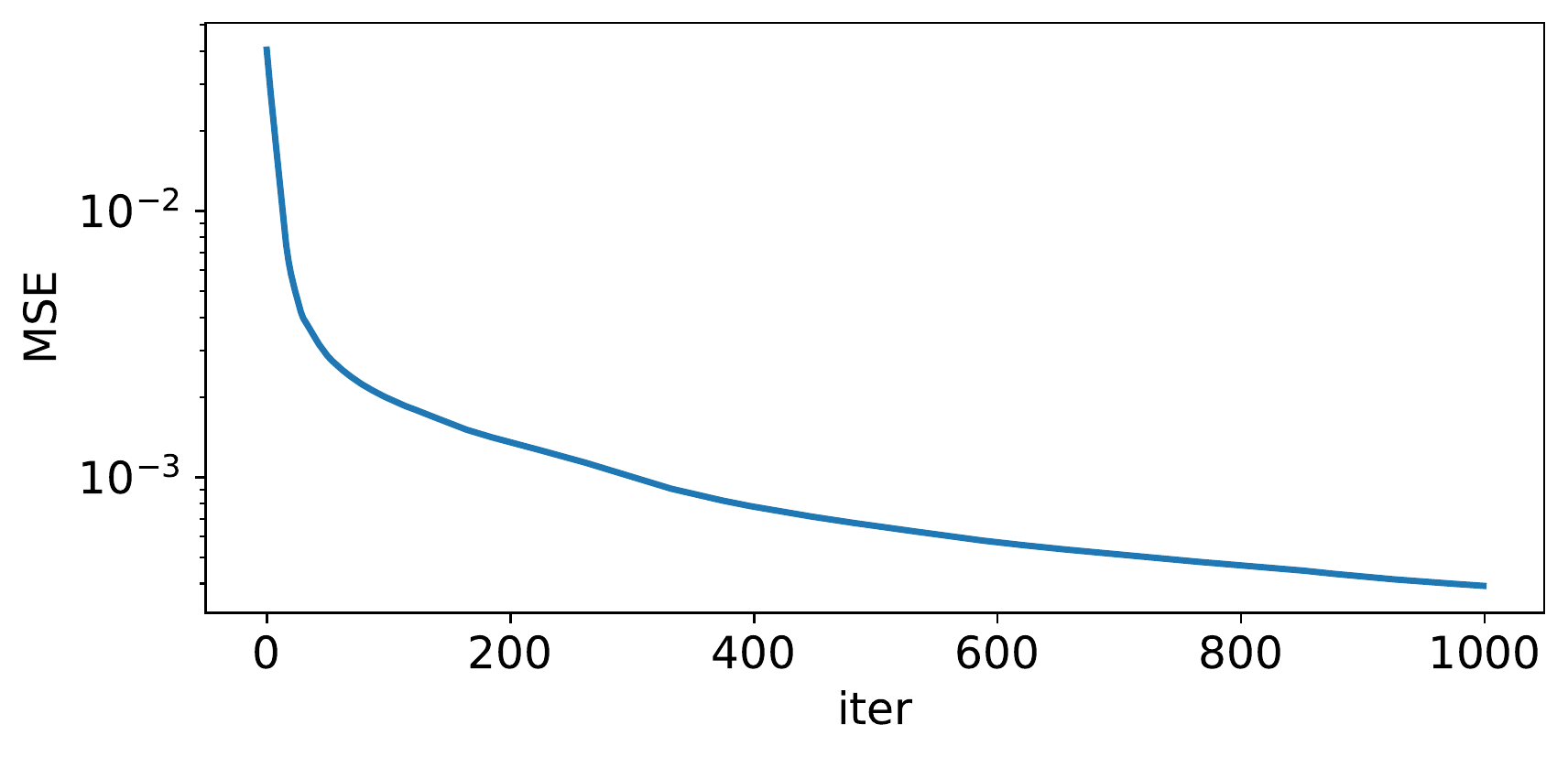}
    \caption{Mixture MSE {vs.}  iteration number}
    \label{fig:mnist-vae-demix-loss-plot}
  \end{subfigure}

  \caption{Demixing results for two deep generative networks. \textbf{(a) \&
      (b)} Ground-truth images used to create \textbf{(c)} the
    mixture. \textbf{(e) \& (f)} The recovered images are shown beside
    \textbf{(g)} the recovered mixture. \textbf{(d)} Final MSE for
    \textbf{(e--g)}, resp. \textbf{(h)} A plot of mixture MSE \emph{vs.}
    iteration number (log-y-scale)}
  \label{fig:mnist-vae-demixing-results}
\end{figure*}


%% file: conclusion.tex
\section{Conclusion}
\label{sec:conclusion}

We prove a sample complexity bound for approximately solving the Lipschitz
demixing problem with high probability on the subgaussian matrix $A$. This
result extends existing work from CS with Lipschitz signals and Gaussian random
matrices to the demixing problem with Lipschitz signals and $K$-subgaussian
matrices. We establish an improved dependence on the subgaussian constant $K$
appearing in the deviation inequality for subgaussian mixing matrices. We
connect the demixing problem for Lipschitz signals to demixing for GNNs, noting
that many popular GNNs are Lipschitz. To support the theory, we provide
numerical simulations for demixing images of handwritten digits. Through these
simulations we suggest a viable algorithm for demixing with trained GNNs, and
demonstrate efficacy of the algorithm's performance on a particular example.

Further analysis into the demixing of Lipschitz signals would be an interesting
subject of future study --- for instance, in the case where the mixing is not
random, or where the networks remain untrained. For example, it would be of both
theoretical and practical interest to analyze signal recovery in the setting
where two signals each in the range of a GNN are mixed \emph{without} the
subgaussian mixing matrix: $b = x^{*} + y^{*} + \eta$. Here, establishing error
bounds for approximate signal recovery may depend on the underlying geometric
structure of the GNNs.

Additionally, theoretical analysis and thorough empirical exploration of our
numerical algorithm remain open questions. For example, whether a convergence
guarantee exists for the algorithm we have proposed, using GNNs, remains
unknown. Further, we have chosen a particular initialization of the algorithm,
but have not analyzed its effect on signal recovery in detail. A comparison of
VAEs' performance on the demixing problem using other datasets would be of
interest, as would a comparison between VAEs and other GNNs such as
GANs. Additionally, it could be worthwhile to validate numerically the effect of
noise on demixing performance. 

Finally, a further investigation may also wish to examine the case where the
GNNs used as the Lipschitz functions $G$ and $H$ themselves have random
weights. Such untrained networks have recently garnered attention as good
structural priors for natural images~\cite{heckel2019deepdecoder,
  heckel2020compressive, lempitsky2018deep}, making them of potential interest
to demixing.







%% file: proofs-gen-demix-random-mapping.tex
\section[Proofs for subgaussian Lipschitz demixing]{Proofs for {$\S$\ref{sec:gen-demix-random-mapping}}}
\label{sec:proofs-gen-demix-random-mapping}

\subsection{Auxiliary results}
\label{sec:auxiliary-results}

\begin{prop}[Gaussian complexity for $|\mathcal{T}| = T$]
  \label{prop:gaussian-complexity-finite-collection}
  Assume that $|\mathcal{T}| = T < \infty$. Observe
  \begin{align*}
    w(\mathcal{T}) %
    & \leq C \sqrt{\log T} \diam(\mathcal{T}); %
    & %
      \frac{1}{3} \left[w(\mathcal{T}) + \rad(\mathcal{T})\right] %
    & \leq \gamma(\mathcal{T}) %
      \leq 2 \left[w(T) + \inf_{y \in \mathcal{T}}\|y\|_{2}\right].
  \end{align*}
\end{prop}

\begin{proof}[Proof of {\autoref{prop:gaussian-complexity-finite-collection}}]
  The bound on $w(\mathcal{T})$ follows from \cite[Exercise
  7.5.10]{vershynin2018high}, and the bounds on $\gamma(\mathcal{T})$ follow
  from \cite[Exercise 7.6.9]{vershynin2018high}.
\end{proof}

\subsection{Proof of {\autoref{lem:berk-8-2}}}
\label{sec:proof-gen-demix-conc}


\begin{proof}[Proof of {\autoref{lem:berk-8-2}}]
  Recall the definition of $X_{z}$ in~\eqref{eq:def-Xz}. Observe that for any
  fixed $z \in \reals^{n}\times \reals^{m}$, $X_{z}$ is a subgaussian random
  variable. For example, if $z \in \sph^{n + m - 1}$, it holds that
  \begin{align*}
    \mathbb{P}\left\{ |X_{z}| \geq C \tilde K t \right\} \leq \exp(-t^{2}).
  \end{align*}
  Consequently, for any $z \in \reals^{n}\times \reals^{m}$,
  \begin{align*}
    \mathbb{P}\left\{ \frac{|X_{z}|}{\sqrt{m} \|z\|_{2}} \geq C \tilde K t \right\} %
    \leq \exp(-m t^{2}).
  \end{align*}
  Substituting $\varepsilon := C t \tilde K$, bounding the probability by a
  quantity $f$ gives
  \begin{align}
    \label{eq:8-2-1}
    \mathbb{P}\left\{ \frac{\|Bz\|_{2}}{\sqrt{m}} \geq (1 + \varepsilon) \|z\|_{2} \right\} %
    &\leq \exp\left(-\frac{c m \varepsilon^{2}}{\tilde K^{2}}\right) = f,
    &
    \varepsilon &= \sqrt{\frac{C \tilde K^{2}}{m} \log\left(\frac{1}{f}\right)}.
  \end{align}

  Next, our goal is to construct a chain of $\delta_{i}$-nets on
  $F(\mathbb{B})$. We do this by constructing a chain of nets on $\mathbb{B}$,
  $M := \{M_{i}\}_{i=0}^{\ell}$, for which each net $M_{i}$ will be a product of
  a net $M_{i}^{(1)}$ on $B^{k}(r)$ and $M_{i}^{(2)}$ on $B^{k'}(r')$. In
  particular, for $B^{k}(r), B^{k'}(r')$ \cite[Corollary
  4.2.13]{vershynin2018high} gives the existence of $\eta_{i}$-nets
  $M_{i}^{(1)}$ and $\eta_{i}'$-nets $M_{i}^{(2)}$ respectively satisfying
  \begin{align*}
    \log |M_{i}^{(1)}| &\leq k \log( 1 + \frac{2r}{\eta_{i}}),
    &
      \log |M_{i}^{(2)}| &\leq k' \log( 1 + \frac{2r'}{\eta_{i}'}).
  \end{align*}
  Choose:
  \begin{align*}
    \eta_{i} := \frac{\delta_{i}}{\sqrt{2}L_{G}}, %
    \eta_{i}' := \frac{\delta_{i}}{\sqrt{2}L_{H}}, %
    \quad\text{where}\quad %
    \delta_{i} := \delta_{0} / 2^{i}, i \in [\ell],
  \end{align*}
  and where $\delta_{0} > 0$ is any number satisfying
  $\delta_{0} < c\delta \tilde K^{-1}$ for some absolute constant $c > 0$.
  Then,
  \begin{align*}
    \log|M_{i}| %
    &= \log |M_{i}^{(1)}| + \log |M_{i}^{(2)}| %
    \leq k \log\left(\frac{2^{i}crL_{G}}{\delta_{0}}\right) %
    + k' \log\left( \frac{2^{i}cr'L_{H}}{\delta_{0}}\right)
    \\
    &= i(k+k')\log 2 + k \log\left(\frac{crL_{G}}{\delta_{0}}\right) %
    + k' \log\left( \frac{cr'L_{H}}{\delta_{0}}\right)
  \end{align*}
  Observe that $N_{i} := F(M_{i})$ is a $\delta_{i}$-net for $F(\mathbb{B})$:
  for any $(u,v) \in \mathbb{B}$ there exists $(u', v') \in M_{i}$ with
  \begin{align*}
    \|F(u, v) - F(u', v')\|_{2} %
    & 
    = \sqrt{
    \|G(u) - G(u')\|_{2}^{2} + \|H(v)-H(v')\|_{2}^{2}
    }
    \\
    &\leq \sqrt{
    L_{G}^{2}\|u - u'\|_{2}^{2} + L_{H}\|v-v'\|_{2}^{2}
    }
    \leq 
      \delta_{i}
  \end{align*}

  Next, define
  \begin{align*}
    T_{i} := \{ z_{i+1} - z_{i} : z_{i+1} \in N_{i+1}, z_{i} \in N_{i}\}
  \end{align*}
  and observe that $T_{i}$ has cardinality satisfying
  \begin{align*}
    \log |T_{i}| %
    &\leq \log(|N_{i+1}||N_{i}|) %
      \leq \log(|M_{i+1}||M_{i}|) %
    \\
    &\leq C(2i+1)(k+k') + 2k \log\left(\frac{crL_{G}}{\delta_{0}}\right) %
    + 2k' \log\left( \frac{cr'L_{H}}{\delta_{0}}\right).
  \end{align*}

  Eventually, for a collection $(\varepsilon_{i}, f_{i})$ connected as
  in~\eqref{eq:8-2-1}, we will union bound to obtain
  \begin{align}
    \label{eq:8-2-2}
    \mathbb{P}\left\{\frac{\|B \zeta \|}{\sqrt m} %
    \leq (1 + \varepsilon_{i}) \|\zeta\|,%
    \forall i,\forall \zeta \in T_{i}\right\} %
    \geq 1 - \sum_{i=0}^{\ell - 1} |T_{i}| f_{i}. 
  \end{align}
  We start by ensuring the right-hand side of the above expression is
  small. Using
  \begin{align}
    \log(|T_{i}|f_{i}) %
    &= \log|T_{i}| + \log (f_{i}) %
      \nonumber
    \\
    \label{eq:8-2-3}
    &\leq C(2i+1)(k+k') + 2k \log\left(\frac{crL_{G}}{\delta_{0}}\right) %
    + 2k' \log\left( \frac{cr'L_{H}}{\delta_{0}}\right) + \log(f_{i}),
  \end{align}
  and choosing
  \begin{align*}
    m \geq 2k \log\left(\frac{crL_{G}}{\delta_{0}}\right) %
    + 2k' \log\left( \frac{cr'L_{H}}{\delta_{0}}\right), %
    \qquad
    \log(f_{i}) =  - \left(\frac{4m}{3} + C(3i+1)(k+k')\right),
  \end{align*}
  gives
  \begin{align*}
    \log(|T_{i}|f_{i}) \leq -Ci(k+k') - \frac{m}{3}. 
  \end{align*}
  In particular,
  \begin{align*}
    \sum_{i=0}^{\ell-1} |T_{i}| f_{i} %
    \leq \exp\left(-\frac{m}{3}\right) \sum_{i=0}^{\ell-1} \exp\left(-Ci(k+k')\right) %
    \leq \exp\left(-\frac{m}{3}\right) \frac{1}{1 - e^{-C}} %
    \leq C \exp\left(-\frac{m}{3}\right).
  \end{align*}

  With the present choices, the expression for $\varepsilon_{i}$ becomes, using
  \eqref{eq:8-2-1},
  \begin{align*}
    \varepsilon_{i} %
    & = \sqrt{\frac{C \tilde K^{2}}{m} \log\left(\frac{1}{f_{i}}\right)} %
      = C \tilde K \sqrt{ 1 + \frac{(3i+1)(k+k')}{m}}. %
  \end{align*}

  Now for any $z \in F(\mathbb{B})$, we can write
  \begin{align*}
    z = z_{0} + (z_{1} - z_{0}) + \cdots + (z_{\ell} - z_{\ell - 1}) + z^{f}
  \end{align*}
  where $z_{i} \in N_{i}, z^{f} := z - z_{\ell}$. Using \eqref{eq:8-2-2}, since
  $\zeta_{i} = z_{i+1} - z_{i} \in T_{i}$, it holds with probability at least
  $1 - C \exp(-m/3)$ that
  \begin{align*}
    m^{-1/2} \sum_{i=0}^{\ell - 1} \|B\zeta_{i}\|_{2} %
    &\leq \sum_{i=0}^{\ell - 1} (1 + \varepsilon_{i}) \|\zeta_{i}\|_{2} %
      \leq \delta_{0} \sum_{i=0}^{\ell - 1} 2^{-i} (1 + \varepsilon_{i}) %
    \\
    & \leq \delta_{0} \sum_{i=0}^{\ell - 1} 2^{-i} %
      \left(1 + C\tilde K \sqrt{1 + \frac{(3i+1)(k+k')}{m}}\right) %
    \\
    & \leq C \delta_{0} \tilde K \sum_{i=0}^{\ell - 1} \frac{1 + \sqrt{3i + 1}}{2^{i}} %
      \leq C \delta_{0} \tilde K %
      < \frac{\delta}{2}. 
  \end{align*}
  Also observe that $\|z^{f}\|_{2} \leq \delta_{0} / 2^{\ell}$ by construction.

  Now, using~\autoref{thm:chen-1-new} with $\mathcal{T} := \sph^{n+m-1}$ we have
  \begin{align*}
    \sup_{z \in \mathcal{T}} \left|\frac{\|B z\|_{2}}{\sqrt m} - 1 \right| %
    \leq C \tilde K \left[\sqrt{\frac{n + m}{m}} + 1\right]
  \end{align*}
  with probability at least $1 - \exp(-m)$. In particular, with the stated
  probability,
  \begin{align*}
    \left\|\frac{B}{\sqrt m}\right\| \leq C\tilde K \left(2 + \sqrt{\frac{n + m}{m}}\right).
  \end{align*}
  Therefore, with probability at least $1 - \exp(-m)$ we have 
  \begin{align*}
    m^{-1/2} \|B z^{f}\|_{2} %
    \leq C\delta_{0}\tilde K 2^{-\ell} \left(2 + \sqrt{\frac{n + m}{m}}\right) %
    \leq C\delta_{0}\tilde K %
    < \frac{\delta}{2},
  \end{align*}
  achieved by taking $\ell \geq C \log(2 + \kappa^{-1})$ where $\kappa := m / n$
  is the aspect ratio of $A$.

  Combining the above results, and noting that it is possible to choose
  $z' = z_{0}$, it holds with probability at least $1 - \exp(-cm)$ that
  \begin{align*}
    m^{-1/2}\|B(z - z')\|_{2} %
    = m^{-1/2}\|B(z - z_{0})\|_{2} %
    \leq m^{-1/2}\|B z^{f}\|_{2} + m^{-1/2} \sum_{i=0}^{\ell - 1} \|B\zeta_{i}\|_{2} %
    < \delta.
  \end{align*}
\end{proof}

\subsection{Proof of {\autoref{lem:s-rec}}}
\label{sec:proof-lem-s-rec}


\begin{proof}[Proof of {\autoref{lem:s-rec}}]
  Observe that there exists a net $M := M^{(1)}\times M^{(2)}$ for $\mathbb{B}$
  where $M^{(1)}$ is a $\frac{\delta}{\sqrt 2 L_{G}}$-net for $B^{k}(r)$ and
  $M^{(2)}$ is a $\frac{\delta}{\sqrt 2 L_{H}}$-net for $B^{k'}(r')$, with
  \begin{align*}
    \log|M| %
    \leq \log |M^{(1)}| + \log|M^{(2)}| %
    \leq k\log\left(1 + \frac{2\sqrt 2 L_{G} r}{\delta}\right) %
    + k'\log\left(1 + \frac{2\sqrt 2 L_{H} r'}{\delta}\right). %
  \end{align*}
  By construction, $F(M) = \{(G(u), H(v)) : (u, v) \in M\}$ is a $\delta$-net
  for $F(\mathbb{B})$. In particular, for any $z, z' \in F(\mathbb{B})$ there
  are $z_{1}, z_{2} \in F(M)$ such that
  \begin{align}
    \label{eq:1}
    \|z - z'\|_{2} %
    \leq \|z - z_{1} \|_{2} + \|z_{1} - z_{2} \|_{2} + \|z_{2} - z' \|_{2} %
    \leq \|z_{1} - z_{2}\|_{2} + 2 \delta.
  \end{align}
  Moreover, for these points, applying \autoref{lem:berk-8-2} for a possibly
  larger net $M$ gives
  \begin{align}
    \label{eq:2}
    \|B(z_{1} - z_{2}) \|_{2} %
    \leq \|B(z_{1} - z) \|_{2} + \|B(z - z') \|_{2} + \|B(z' - z_{2}) \|_{2} %
    \leq \|B(z - z') \|_{2} + C\delta\sqrt m.
  \end{align}

  Next, defining
  \begin{align*}
    \mathcal{T} := \left\{ \frac{z_{1} - z_{2}}{\|z_{1} - z_{2}\|_{2}} : z_{1}, z_{2} \in F(M) \right\}, 
  \end{align*}
  and noting $\gamma(\mathcal{T}) \leq C \sqrt{\log |M|}$
  by~\autoref{prop:gaussian-complexity-finite-collection}, we may
  use~\autoref{lem:subgaussian-process-better-constant} to obtain
  \begin{align*}
    \sup_{z_{1}, z_{2} \in F(M)}\left|%
    \frac{\|B(z_{1} - z_{2})\|_{2}}{\sqrt m \|z_{1} - z_{2}\|_{2}} - 1\right| %
    \leq C\tilde K \sqrt{\frac{\log|M|}{m}}, %
  \end{align*}
  with probability at least $1 - \exp(-m)$. Accordingly, one has
  \begin{align}
    \label{eq:3}
    (1 - \alpha)\|z_{1} - z_{2}\|_{2} %
    \leq m^{-1/2}\|B(z_{1}-z_{2})\|_{2} %
    \leq (1 + \alpha) \|z_{1}-z_{2}\|_{2}
  \end{align}
  with probability at least $1 - \exp(-m)$ provided that
  \begin{align}
    \label{eq:4}
    m %
    \geq \frac{C\tilde K^{2}}{\alpha^{2}}(k + k')\log\left(%
    \frac{\lambda\tilde K}{\delta}\right) %
    \geq \frac{C\tilde K^{2}}{\alpha^{2}} \log |M|, \qquad %
    \lambda := \max\{ L_{G} r, L_{H} r' \}. 
  \end{align}
  Combining (\ref{eq:1}--\ref{eq:3}) under condition \eqref{eq:4} gives, for any $z, z' \in F(\mathbb{B})$,
  \begin{align*}
    (1 - \alpha) \|z - z'\|_{2} %
    & \leq (1-\alpha)\|z_{1} - z_{2}\|_{2} + 2\delta %
      \leq \|m^{-1/2}B(z_{1} - z_{2})\|_{2} + 2\delta %
      \\
    & \leq \|m^{-1/2}B(z - z')\|_{2} + 4\delta. 
  \end{align*}
  with probability at least $1 - \exp(-cm)$.
  
\end{proof}

\begin{rmk}[Simplified sample complexity bound]
  The expression appearing in~\eqref{eq:4},
  \begin{align*}
    m  %
    \geq \frac{C\tilde K^{2}}{\alpha^{2}}(k + k') \log\left(\frac{\lambda \tilde K }{\delta}\right),
  \end{align*}
  is simpler (\ie more aesthetic) than the analogous condition appearing in the
  statement of~\autoref{lem:s-rec}, but also a stronger assumption on $m$, due
  to the definition of $\lambda$.
\end{rmk}

\subsection{Proof of {\autoref{thm:gen-demix-whp}}}
\label{sec:proof-thm-gen-demix-random}

  

Before proving \autoref{thm:gen-demix-whp}, we modify \cite[Lemma 4.3]{bora2017compressed} so that the statement fits the present setting.
\begin{lem}[modified {\cite[Lemma 4.3]{bora2017compressed}}]
  \label{lem:bora-4-3-alt}
  Let $\tilde B := \bmat{rr}{\tilde A & I_{m}}$ where
  $\tilde A \in \mathbb{R}^{m \times n}$ is drawn from a distribution such that:
 \begin{enumerate}[itemsep=0pt, topsep=0pt]
 \item $\tilde B$ satisfies $\SREC(S, \gamma, \delta)$ with probability $1-p$;
 \item $\|\tilde B z\|_{2} \leq C\tilde K \|z \|_{2}$ for every fixed
   $z \in \reals^{n}\times\reals^{m}$ with probability $1-p$.
 \end{enumerate}
 For any $z^{*} \in \reals^{n}\times\reals^{m}$ and noise $\eta \in \reals^{m}$,
 let $b = \tilde B z^{*} + \eta$. Let $\hat{z}$ satisfy
 \begin{align*}
   \| b - \tilde B \hat{z}\|_{2} \leq \min _{z \in S} \|b - \tilde B z\|_{2} + \epsilon.
 \end{align*}
 Then, with probability $1 - 2 p$,
 \begin{align*}
   \left\|\hat{z} - z^{*}\right\|_{2} %
   \leq \left( \frac{C\tilde K}{\gamma}+1\right) \min _{z \in S} \left\| z^{*} - z \right\|_{2} %
   + \frac{1}{\gamma} ( 2\|\eta\|_{2} + \epsilon + \delta ).
 \end{align*}

\end{lem}

\begin{proof}[{Proof of \autoref{lem:bora-4-3-alt}}]
  We proceed as in the proof of \cite[Lemma 4.3]{bora2017compressed}. Define
  $\bar z \in \argmin_{z \in S} \left\| z^{*} - z \right\|_{2}$. Since $\tilde B$
  satisfies $\SREC(S, \gamma, \delta)$, with probability $1-p$, it holds with
  probability $1-p$, using the assumption on $\hat z$, that
  \begin{align*}
    \|\bar z - \hat z\|_{2} %
    &\leq \gamma^{-1}\left( \|\tilde B \bar z - b\|_{2} + \|\tilde B \hat z - b\|_{2} + \delta\right) %
    \leq \gamma^{-1}\left( \|\tilde B \bar z - b\|_{2} + \min_{z \in S} \|\tilde B z - b\|_{2} + \delta + \varepsilon \right) %
    \\
    &\leq \gamma^{-1}\left(2 \|\tilde B \bar z - b\|_{2} + \delta + \varepsilon \right)
    \leq \gamma^{-1}\left(2 \|\tilde B (\bar z - z^{*}) \|_{2} + \delta + \varepsilon + 2 \|\eta\|_{2} \right)
  \end{align*}
  By assumption, it also holds with probability at least $1-p$ that
  $\|\tilde B(\bar z - z^{*})\|_{2} \leq C\tilde K \|\bar z - z^{*}\|_{2}$. Therefore,
  with probability at least $1 - 2p$,
  \begin{align*}
    \|z^{*} - \hat z\|_{2} %
    & \leq \|z^{*} - \bar z\|_{2} + \|\bar z - \hat z\|_{2}
    \\
    & \leq \|z^{*} - \bar z\|_{2} + \gamma^{-1}\left(2 \|\tilde B (\bar z - z^{*}) \|_{2} + \delta + \varepsilon + 2 \|\eta\|_{2} \right) %
    \\
    & \leq \|z^{*} - \bar z\|_{2} + \gamma^{-1}\left(C\tilde K \|\bar z - z^{*} \|_{2} + \delta + \varepsilon + 2 \|\eta\|_{2} \right) %
    \\
    & = \left(1 + \frac{C\tilde K}{\gamma}\right) \|z^{*} - \bar z\|_{2} + \frac{\delta + \varepsilon + 2 \|\eta\|_{2}}{\gamma} %
  \end{align*}
\end{proof}

\begin{proof}[Proof of {\autoref{thm:gen-demix-whp}}]
  Suppose $\alpha \in (0, 1)$. By definition of $A$ and
  $B := \bmat{rr}{A & \sqrt{m} I_{m}}$, the normalized mixing matrix
  $\tilde B := m^{-1/2} B$ satisfies $\SREC(F(\mathbb{B}), 1 - \alpha, \delta)$
  with probability at least $1 - \exp(-cm)$ when
  \begin{align*}
    m %
    \geq \frac{C \tilde K^{2}}{\alpha^{2}} \left(k \log \left(\frac{L_{G}r}{\delta}\right) %
    + k' \log \left(\frac{L_{H}r'}{\delta}\right)\right). 
  \end{align*}
  Observe that if $\lambda := \max\{ L_{G}r, L_{H}r'\}$ then
  $m \geq \frac{C\tilde K^{2}}{\alpha^{2}}(k + k')
  \log\left(\frac{\lambda}{\delta}\right)$ implies the above
  condition. Furthermore, subgaussianity of $X_{z}$ gives for any fixed
  $z \in \reals^{n}\times \reals^{m}$,
  \begin{align*}
    \left| \|\tilde Bz\|_{2} - \|z\|_{2}\right| %
    \leq C \alpha \tilde K \|z\|_{2} 
  \end{align*}
  with probability at least $1 - \exp(-\alpha^{2} m)$. Choosing
  $\alpha := \frac{1}{2}$ and applying \autoref{lem:bora-4-3-alt} gives
  \begin{align*}
    \|\hat z - z^{*} \|_{2} %
    \leq C\tilde K \min_{z \in F(\mathbb{B})} \|z^{*} - z\|_{2} %
    + C(2\|\eta\|_{2} + \delta + \varepsilon). 
  \end{align*}
\end{proof}


%% file: proofs-improved-subgaussian-constant.tex
\section[Proofs for improved subgaussian constant]{Proofs for {$\S$\ref{sec:impr-k-depend}}}
\label{sec:impr-subg-const}

\subsection{Proof of {\autoref{lem:subgaussian-process-better-constant}}}
\label{sec:proof-improved-increments}

It should be noted that the proof
of~\autoref{lem:subgaussian-process-better-constant} strongly resembles the
proof of an analogous result, \cite[Lemma 5]{chen2018stable}. We include full
details here for the sake of clarity to the reader, and note that the proof
requires only a handful of changes to leverage an improved Bernstein's
inequality,~\cite[New Bernstein's inequality]{jeong2019non}.

\begin{thm*}[{\cite[New Bernstein's inequality]{jeong2019non}}]
  \label{thm:new-bernstein}
  Let $a = (a_{1}, \ldots, a_{m})$ be a fixed nonzero vector and let
  $Y_{1} \, \ldots, Y_{m}$ be independent, centered subexponential random
  vectors satisfying $\E |Y_{i}| \leq 2$ and
  $\|Y_{i}\|_{\psi_{1}} \leq K_{i}^{2}$ with $K_{i} \geq 2$. Then, for every
  $u \geq 0$,
  \begin{align*}
    \mathbb{P}\left\{ \left|\sum_{i=1}^{m} a_{i} Y_{i}\right| \geq u \right\} %
    \leq 2 \exp \left[- c \cdot \min \left(\frac{u^{2}}{\sum_{i=1}^{m} a_{i}^{2} K_{i}^{2} \log K_{i}}, %
    \frac{u}{\|a\|_{\infty} K^{2}\log K} \right)\right]
  \end{align*}
  where $K := \max_{i} K_{i}$ and $c$ is an absolute constant.
\end{thm*}

We will also rely on a version of Hoeffding's inequality for weighted sums of
subgaussian random variables. We state this theorem as it appears
in~\cite[Theorem 2.6.3]{vershynin2018high}.

\begin{thm*}[{\cite[Theorem 2.6.3]{vershynin2018high}}]
  \label{thm:hoeffding}
  Let $X_{1}, \ldots, X_{m}$ be independent, centered, subgaussian random
  variables, and $a =\left(a_{1}, \ldots, a_{m}\right) \in \reals^{m}$. Then,
  for every $t \geq 0$, we have
  \begin{align*}
    \mathbb{P}\left\{%
    \left|\sum_{i=1}^{m} a_{i} X_{i}\right| \geq t\right\} %
    \leq 2 \exp \left(-\frac{c t^{2}}{K^{2}\|a\|_{2}^{2}}\right)
  \end{align*}
  where $K=\max _{i}\left\|X_{i}\right\|_{\psi_{2}}$.
\end{thm*}

As in the proof of \cite[Lemma~5]{chen2018stable}
or~\cite[Theorem~1.3]{liaw2017simple}, we start by proving the result for
$z \in \sph^{n + m - 1}$ and $z' = 0$; then for $z, z' \in \sph^{n + m - 1}$;
and finally we prove the incremement inequality for any
$z, z' \in \reals^{n}\times \reals^{m}$.

\begin{lem}
  \label{lem:improved-constant-case-1}
  Suppose $A \in \reals^{m \times n}$ is a $K$-subgaussian matrix with
  $K \geq 2$. Then,
  \begin{align*}
    \left\|\|Ax + \sqrt m y \|_{2} - \sqrt m \right\|_{\psi_{2}} %
    \leq C K \sqrt{\log K}, %
    \quad \text{for every } \bmat{r}{x\\y} \in \sph^{n + m - 1}
  \end{align*}
\end{lem}

\begin{proof}[Proof of {\autoref{lem:improved-constant-case-1}}]
  For any $t \geq 0$ we have
  \begin{align*}
    p %
    & := \mathbb{P}\left\{ \left|\frac{1}{m} \|Ax + \sqrt m y \|_{2} - 1\right| \geq t \right\} %
      = \mathbb{P}\left\{\left|\frac{1}{m} \|Ax\|_{2}^{2} - \|x\|_{2}^{2} + \frac{2}{\sqrt m}  \ip{Ax, y}\right| \geq t \right\} %
      \\
    & \leq \mathbb{P}\left\{\left|\frac{1}{m} \|Ax\|_{2}^{2} - \|x\|_{2}^{2}\right| %
      + \left|\frac{2}{\sqrt m} \ip{Ax, y}\right| \geq t \right\} %
      \\
    & \leq \mathbb{P}\left\{\left|\frac{1}{m} \|Ax\|_{2}^{2} - \|x\|_{2}^{2}\right| \geq \frac{t}{2}\right\} %
      + \mathbb{P}\left\{\left|\frac{2}{\sqrt m} \ip{Ax, y}\right| \geq \frac{t}{2} \right\} %
      =: p_{1} + p_{2}
  \end{align*}
  To bound $p_{1}$, write
  \begin{align*}
    \frac{1}{m} \|Ax\|_{2}^{2} - \|x\|_{2}^{2} %
    = \frac{1}{m} \sum \left[\ip{A_{i}^{T}, x}^{2} - \|x\|_{2}^{2}\right] %
    =: \frac{1}{m} \sum_{i=1}^{m} Z_{i}
  \end{align*}
  It holds by assumption that the collection $\{\ip{A_{i}^{T}, x}\}_{i\in [m]}$
  is independent, and that each are centered and subgaussian random vectors with
  \begin{align*}
    \E \ip{A_{i}^{T}, x}^{2} = \|x\|_{2}^{2} \leq 1 %
    \quad \text{and} \quad %
    \|\ip{A_{i}^{T}, x}\|_{\psi_{2}} \leq K \|x\|_{2} \leq K. 
  \end{align*}
  Standard properties of subexponential and subgaussian random
  variables~\cite[Chapter~2]{vershynin2018high} thus yield:
  \begin{align*}
    \|Z_{i}\|_{\psi_{1}} %
    = \left\| \ip{A_{i}^{T}, x}^{2} - \E \ip{A_{i}^{T}, x}^{2} \right\|_{\psi_{1}} %
    \leq C \left\| \ip{A_{i}^{T}, x}^{2}\right\| %
    \leq C \left\| \ip{A_{i}^{T}, x}\right\|^{2} %
    \leq C K^{2}.
  \end{align*}
  In particular, the assumptions of \nameref{thm:new-bernstein} are
  satisfied. With $a = \left(\frac{1}{m}, \ldots, \frac{1}{m}\right)$,
  \begin{align*}
    p_{1} %
    = \mathbb{P}\left\{\left|\frac{1}{m} \sum_{i=1}^{m} Z_{i} \right| \geq \frac{t}{2} \right\} %
    & \leq 2 \exp \left[-c \cdot \min \left(\frac{m t^{2}}{C K^{2} \log K}, %
      \frac{m t}{C K^{2} \log K} \right)\right] %
    \\
    & \leq 2 \exp \left[ - \frac{c_{1} m}{K^{2} \log K} \cdot \min \left(t^{2}, t\right)\right]. 
  \end{align*}

  To bound $p_{2}$, we write
  $\ip{Ax, y_{i}} = \sum_{i=1}^{m} y_{i} \ip{A_{i}^{T}, x}$ and apply
  Hoeffding's inequality, noting that
  $\|\ip{A_{i}^{T}, x} \|_{\psi_{2}} \leq K$:
  \begin{align*}
    p_{2} %
    = \mathbb{P}\left\{ \left|\frac{1}{\sqrt{m}} \sum_{i=1}^{m} y_{i} \ip{A_{i}^{T}, y}\right| %
    \geq \frac{t}{4}\right\} %
    \leq 2 \exp\left[- \frac{-cmt^{2}}{K^{2} \|y\|_{2}^{2}}\right] %
    \leq 2 \exp \left[ - \frac{c_{2} m t^{2}}{K^{2} \log K} \right].
  \end{align*}
  Here, we have used that $K \geq 2$ to obtain the right-hand
  expression. Combining the two bounds gives
  \begin{align*}
    p %
    \leq 2 \exp \left[ - \frac{c_{1} m}{K^{2} \log K} \cdot \min \left(t^{2}, t\right)\right] %
    + 2 \exp \left[ - \frac{c_{2} m t^{2}}{K^{2} \log K} \right] %
    \leq 4 \exp \left[ - \frac{c_{0} m}{K^{2}\log K} \min \left( t^{2}, t \right) \right].
  \end{align*}
  Next, we establish a concentration inequality for
  $\frac{1}{\sqrt m} \|A x + \sqrt m y \|_{2} - 1$. Observe that
  \begin{align*}
    |z - 1| > \delta %
    \quad \implies \quad %
    |z^{2} - 1| \geq \max \left( \delta^{2}, \delta \right) %
    \quad \text{for any } z\geq 0, \delta \geq 0. 
  \end{align*}
  In particular, we set $t := \max \left\{ \delta^{2}, \delta \right\}$ to obtain
  for any $\delta \geq 0$,
  \begin{align*}
    \mathbb{P}\left\{ \left| \frac{1}{\sqrt m} \|Ax + \sqrt m y \|_{2} - 1\right| \geq \delta \right\} %
    \leq \mathbb{P}\left\{ \left|\frac{1}{m}\left\| Ax + \sqrt m y \right\|_{2}^{2} - 1\right| \geq \max \left( \delta^{2}, \delta \right) \right\} %
    \leq 4 \exp \left[ - \frac{c_{0} m \delta^{2}}{K^{2}\log K} \right]. 
  \end{align*}
  In particular,
  \begin{align*}
    \left\| \left\| Ax + \sqrt m y \right\|_{2} - \sqrt m \right\|_{\psi_{2}} %
    \leq C K \sqrt{\log K} %
    \quad \text{for any } (x, y) \in \sph^{n + m - 1}
  \end{align*}
\end{proof}

\begin{lem}
  \label{lem:improved-constant-case-2}
  Suppose $A \in \reals^{m \times n}$ is a $K$-subgaussian matrix with
  $K \geq 2$. Then,
  \begin{align*}
    \left\| \left\| Ax + \sqrt m y \right\|_{2} - \left\| A x' + \sqrt m y' \right\|_{2} \right\|_{\psi_{2}}
    \leq C K \sqrt{\log K} \sqrt{\left\| x - x' \right\|_{2}^{2} + \left\| y - y' \right\|_{2}^{2}}, %
  \end{align*}
  for every $\bmat{r}{x\\y}, \bmat{r}{x'\\y'} \in \sph^{n + m - 1}$.
\end{lem}

\begin{proof}[Proof of {\autoref{lem:improved-constant-case-2}}]
  It is enough to show that for every $t \geq 0$
  \begin{align*}
    p %
    := \mathbb{P}\left\{%
    \frac{\left|\left\| Ax + \sqrt m y \right\|_{2} - \left\| Ax' + \sqrt m y' \right\|_{2}\right| %
    }{\sqrt{\left\| x-x' \right\|_{2}^{2} + \left\| y - y' \right\|_{2}^{2}}} %
    \geq t\right\} %
    \leq C \exp \left[ - \frac{ct^{2}}{K^{2}\log K} \right].
  \end{align*}
  As in \cite[Lemma 7]{chen2018stable} and \cite{liaw2017simple}, we proceed
  differently for small and large $t$.

  \textbf{First assume} that $t \geq 2 \sqrt m$ and let
  \begin{align*}
    u &:= \frac{x - x'}{\sqrt{\| x - x' \|_{2}^{2} + \|y-y'\|_{2}^{2}}}, %
      &
    v &:= \frac{y - y'}{\sqrt{\|x-x'\|_{2}^{2} + \|y-y'\|_{2}^{2}}}. 
  \end{align*}
  An application of triangle inequality gives
  \begin{align*}
    p %
    & \leq \mathbb{P}\left\{%
    \frac{\left|\left\| A(x-x') + \sqrt m (y-y') \right\|_{2} \right| %
    }{\sqrt{\left\| x-x' \right\|_{2}^{2} + \left\| y - y' \right\|_{2}^{2}}} %
    \geq t\right\} %
      = \mathbb{P}\left\{ \|Au - \sqrt m v \|_{2} \geq t \right\} %
  \end{align*}
  Subtracting $\sqrt m$ from both sides of the relation describing the event,
  then using the assumption $t \geq 2 \sqrt m$ gives
  \begin{align*}
    \mathbb{P}\left\{ \|Au - \sqrt m v \|_{2} \geq t \right\}
    & = \mathbb{P}\left\{ \|Au - \sqrt m v \|_{2} - \sqrt m \geq t - \sqrt m \right\} %
    \\
    & \leq \mathbb{P}\left\{ \|Au - \sqrt m v \|_{2} - \sqrt m \geq \frac{t}{2} \right\} %
    \leq 2 \exp \left[ - \frac{c t^{2}}{K^{2} \log K} \right],
  \end{align*}
  where the latter line follows from \autoref{lem:improved-constant-case-1}.

  \textbf{Next assume} that $t \leq 2 \sqrt m$ and define
  \begin{align*}
    u' &:= x + x', & v' := y + y'.
  \end{align*}
  Observe that multiplication by
  $\|Ax + \sqrt m y \|_{2} + \|Ax' + \sqrt m y'\|_{2}$ gives
  \begin{align*}
    p %
    &= \mathbb{P}\left\{\left|%
    \frac{\|Ax + \sqrt m y \|_{2}^{2} - \|A x' + \sqrt m y'\|_{2}^{2}}{\sqrt{\|x-x'\|_{2}^{2} + \|y-y'\|_{2}^{2}}}%
    \right| %
    \geq t \left( \|Ax + \sqrt m y \|_{2} + \|Ax' + \sqrt m y'\|_{2} \right)%
      \right\} %
    \\
    & \leq \mathbb{P}\left\{
      \left|
      \frac{\ip{A(x-x') + \sqrt m (y - y'), A(x + x') + \sqrt m (y + y')}}{\sqrt{\|x-x'\|_{2}^{2} + \|y-y'\|_{2}^{2}}}
      \right|
      \geq t \cdot \|A x + \sqrt m y\|_{2}
      \right\} %
    \\
    & = \mathbb{P}\left\{
      \left|
      \ip{Au + \sqrt m v, Au' + \sqrt m v'}
      \right|
      \geq t \cdot \|A x + \sqrt m y\|_{2}
      \right\}.
  \end{align*}
  Next, define the following events
  \begin{align*}
    \Omega_{0} %
    &:= \left\{
      \left|
      \ip{Au + \sqrt m v, Au' + \sqrt m v'}
      \right|
      \geq t \cdot \|A x + \sqrt m y\|_{2}
      \right\}, %
    &
      \Omega_{1} %
    &:= \left\{ \|Ax + \sqrt m y \|_{2} \geq \frac{\sqrt m}{2} \right\}.
  \end{align*}
  By the law of total probability,
  \begin{align*}
    p %
    &\leq \mathbb{P} \left( \Omega_{0} \right) %
    = \mathbb{P} \left( \Omega_{0} \mid \Omega_{1} \right) \cdot \mathbb{P}\left\{\Omega_{1}\right\} %
      + \mathbb{P} \left( \Omega_{0} \mid \Omega_{1}^{C} \right) \cdot \mathbb{P}\left\{\Omega_{1}^{C}\right\} %
    \\
    &\leq \mathbb{P}\left\{\Omega_{0} \And \Omega_{1}\right\} + \mathbb{P}\left\{\Omega_{1}^{C}\right\} %
    =: p_{1} + p_{2}.
  \end{align*}
  Notice that $p_{2}$ may be bounded using
  \autoref{lem:improved-constant-case-1}. Indeed,
  \begin{align*}
    p_{2} %
    & = \mathbb{P}\left\{\|Ax + \sqrt m y \|_{2} \leq \frac{\sqrt m}{2}\right\} %
      \leq \mathbb{P}\left\{\left|\|Ax + \sqrt m y \|_{2} - \sqrt m \right| \geq \frac{\sqrt m}{2}\right\} %
    \\
    & \leq \mathbb{P}\left\{\left|\|Ax + \sqrt m y \|_{2} - \sqrt m \right| \geq \frac{t}{4}\right\} %
      \leq 2\exp \left[ - \frac{ct^{2}}{K^{2}\log K} \right].
  \end{align*}
  Next we bound $p_{1}$. Observe that
  \begin{align*}
    p_{1} %
    & = \mathbb{P}\left\{\Omega_{0} \And \Omega_{1}\right\} %
      \leq \mathbb{P}\left\{\left|\ip{Au + \sqrt m v, Au' + \sqrt m v'} \right| \geq \frac{t \sqrt m}{2}\right\}
    \\
    & \leq \underbrace{\mathbb{P}\left\{\left|\ip{Au, Au'} + m \ip{v, v'} \right| \geq \frac{t \sqrt m}{4}\right\}}_{p_{1a}}
    \\
    & + \underbrace{\mathbb{P}\left\{\left|\sqrt m \ip{Au, v'} \right| \geq \frac{t \sqrt m}{8}\right\}}_{p_{1b}} %
      + \underbrace{\mathbb{P}\left\{\left|\sqrt m \ip{Au', v} \right| \geq \frac{t \sqrt m}{8}\right\}}_{p_{1c}}. %
  \end{align*}
  Before we start by bounding $p_{1a}$, observe that 
  \begin{align*}
    \ip{v, v'} %
    = \frac{\|y\|_{2}^{2} - \|y'\|_{2}^{2}}{\sqrt{\|x-x'\|_{2}^{2} + \|y - y'\|_{2}^{2}}} %
    = - \frac{\|x\|_{2}^{2} - \|x'\|_{2}^{2}}{\sqrt{\|x-x'\|_{2}^{2} + \|y - y'\|_{2}^{2}}} %
    = - \ip{u, u'}.
  \end{align*}
  Consequently, $p_{1a}$ may be written as
  \begin{align*}
    p_{1a} %
    & = \mathbb{P}\left\{\left|\ip{Au, Au'} + m \ip{v, v'} \right| %
      \geq \frac{t \sqrt m}{4}\right\}
    \\
    & = \mathbb{P}\left\{\left|\ip{Au, Au'} - m \ip{u, u'} \right| %
      \geq \frac{t \sqrt m}{4}\right\}
    \\
    & = \mathbb{P}\left\{\left|\sum_{i=1}^{m} \left[\ip{A_{i}^{T},u}\ip{A_{i}^{T}, u'} %
      - \ip{u, u'}\right] \right| \geq \frac{t \sqrt m}{4}\right\} %
    \\
    & = \mathbb{P}\left\{\left|\sum_{i=1}^{m}Z_{i}\right| \geq \frac{t \sqrt m}{4}\right\}. 
  \end{align*}
  Observe that $Z_{i}, i \in [m]$ are independent, centered subexponential
  random variables satisfying
  \begin{align*}
    \|Z_{i}\|_{\psi_{1}} %
    \leq C \|\ip{A_{i}^{T}, u} \|_{\psi_{2}} \cdot C \|\ip{A_{i}^{T}, u'} \|_{\psi_{2}} %
    \leq C K^{2} \|u\|_{2} \|u'\|_{2} %
    \leq C K^{2},
  \end{align*}
  since $\|u\|_{2}, \|u'\|_{2} \leq 2$. Moreover, observe that
  \begin{align*}
    \E \left|Z_{i}\right| %
    & = \E \left|\ip{A_{i}^{T},u}\ip{A_{i}^{T}, u'} - \ip{u, u'} \right| %
      \leq \frac{1}{2}\E \left[ \ip{A_{i}^{T},u}^{2} + \ip{A_{i}^{T}, u'}^{2} \right] + \ip{u, u'} %
    \\
    & = \frac{1}{2} \left[ \|u\|_{2}^{2} + \|u'\|_{2}^{2}\right] + \ip{u, u'} %
      = \frac{1}{2} \left( \|u\|_{2} + \|u'\|_{2} \right)^{2} = \frac{9}{2}. 
  \end{align*}
  Thus, as per \cite[Remark 2.5]{jeong2019non}, there are absolute constants
  $c, c' > 0$ such that, by \nameref{thm:new-bernstein},
  \begin{align*}
    p_{1a} %
    & \leq 2 \exp \left[ - c' \min \left( \frac{t^{2}}{16 K^{2}\log K}, \frac{t \sqrt m}{4 K^{2}\log K} \right) \right] %
        \leq 2 \exp \left[ - \frac{ c t^{2}}{K^{2}\log K} \right]. %
  \end{align*}
  Above, the latter inequality follows by $t \leq 2 \sqrt
  m$. Applying~\nameref{thm:hoeffding} bounds $p_{1b}$ and $p_{1c}$:
  \begin{align*}
    p_{1b} %
    & \leq 2 \exp \left[ - \frac{ct^{2}}{K^{2}} \right], %
    & %
      p_{1c} %
      & \leq 2 \exp \left[ - \frac{ct^{2}}{K^{2}} \right]. 
  \end{align*}
  Combining $p_{1a}, p_{1b}, p_{1c}$ and $p_{2}$ gives
  \begin{align*}
    p %
    \leq p_{1} + p_{2} %
    \leq p_{1a} + p_{1b} + p_{1c} + p_{2} %
    \leq C \exp \left[ - \frac{ct^{2}}{K^{2}\log K} \right].
  \end{align*}

\end{proof}

\begin{proof}[Proof of {\autoref{lem:subgaussian-process-better-constant}}]
  Without loss of generality assume $\|z\|_{2} = 1$ and $\tau := \|z' \|_{2} \geq 1$. Define $z'' := x' / \tau$. Then
  \begin{align*}
    \|X_{z} - X_{z'} \|_{\psi_{2}}
    \leq \underbrace{\|X_{z} - X_{z''}\|_{\psi_{2}}}_{R_{1}} %
    + \underbrace{\|X_{z''} - X_{z'}\|_{\psi_{2}}}_{R_{2}}
  \end{align*}
  By \autoref{lem:improved-constant-case-2},
  $R_{1} \leq C K\sqrt{\log K} \|z - z''\|_{2}$. By collinearity of $z''$ and
  $z'$, \autoref{lem:improved-constant-case-1} yields:
  \begin{align*}
    R_{2} = \|X_{z''} - X_{z'}\|_{\psi_{2}} %
    = \|z' - z''\|_{2} \cdot \|X_{z''}\|_{\psi_{2}} %
    \leq C K\sqrt{\log K} \cdot \|z' - z''\|_{2}. 
  \end{align*}
  Combining $R_{1}$ and $R_{2}$ gives
  \begin{align*}
    R_{1} + R_{2} %
    \leq C K \sqrt{\log K} \left( \|z - z''\|_{2} + \|z' - z''\|_{2} \right) %
    \leq C K \sqrt{\log K} \cdot \|z - z'\|_{2},
  \end{align*}
  where the final inequality follows from the ``reverse triangle
  inequality''~\cite[Exercise 9.1.7]{vershynin2018high}. Accordingly, for any
  $z, z' \in \reals^{n}\times \reals^{m}$ one has as desired,
  \begin{align*}
    \|X_{z} - X_{z'} \|_{\psi_{2}} \leq C K\sqrt{\log K} \cdot \|z-z'\|_{2}. 
  \end{align*}

\end{proof}

\begin{rmk}
  \label{rmk:new-bernstein-constants}
  In~\autoref{lem:improved-constant-case-1}
  and~\autoref{lem:improved-constant-case-2} it was assumed that $K \geq
  2$. Notably, any $K \geq c_{0}$ for $c_{0} > 1$ would work, leading only to a
  different value for the absolute constant in the final expression. This bound
  on $K$ is necessary to ensure the $\log K$ term is positive
  (\emph{cf}.~\cite[Remark~2.5]{jeong2019non}). Fortunately, as noted at the
  beginning of~\nameref{sec:gen-demix-random-mapping}, the matrices considered
  in this work admit subgaussian constants satisfying
  $K \geq K_{0} = (\log 2)^{-1/2} \approx 1.201$.
\end{rmk}

\subsection{Proof of {\autoref{thm:chen-1-new}}}
\label{sec:proof-chen-new}

Before proving \autoref{thm:chen-1-new}, we introduce a necessary tool: the
majorizing measures theorem.

\begin{thm}[{\cite[Theorem~4.1]{liaw2017simple}}]
  \label{thm:majorizing-measures}
  Consider a random process $(X_{z})_{z \in \mathcal{T}}$ indexed by points $z$
  in a bounded set $\mathcal{T} \subseteq \reals^{n}$. Assume that the process
  has sub-gaussian incremements: that there exists $M \geq 0$ such that
  \begin{align*}
    \|X_{z} - X_{z'}\|_{\psi_{2}} \leq M \|z - z'\|_{2}%
    \quad\text{for every } z,z' \in \mathcal{T}. 
  \end{align*}
  Then,
  \begin{align*}
    \E \sup_{z, z' \in \mathcal{T}} \left| X_{z} - X_{z'} \right| %
    \leq C M \w(\mathcal{T}).
  \end{align*}
  Moreover, for any $u \geq 0$, the event
  \begin{align*}
    \sup_{z, z' \in \mathcal{T}} \left| X_{z} - X_{z'} \right| %
    \leq C M \left[ \w(\mathcal{T}) + u\cdot \diam(\mathcal{T}) \right]
  \end{align*}
  holds with probability at least $1 - e^{-u^{2}}$, where
  $\diam(\mathcal{T}) := \sup_{z, z' \in \mathcal{T}} \|z - z'\|_{2}$.
\end{thm}

This phrasing of the theorem appears in~\cite[Theorem~4.1]{liaw2017simple}
and~\cite[II.B.3]{chen2018stable} alike. As those authors note, the expectation
bound of this theorem can be found
in~\cite[Theorem~2.1.1,~2.1.5]{talagrand2006generic}; the high-probability bound
in~\cite[Theorem~3.2]{dirksen2015tail}.

\begin{proof}[Proof of {\autoref{thm:chen-1-new}}]
  Combining~\autoref{lem:subgaussian-process-better-constant}
  and~\nameref{thm:majorizing-measures} gives, for $X_{z}$ as defined
  in~\eqref{eq:def-Xz},
  \begin{align*}
    \E \sup_{z, z' \in \mathcal{T}} |X_{z} - X_{z'} | \leq C \tilde K \w(\mathcal{T}). 
  \end{align*}
  Therefore, an application of triangle inequality gives
  \begin{align*}
    \E \sup_{z \in \mathcal{T}} |X_{z}| \leq C\tilde K \w(\mathcal{T}) + \E|X_{z_{0}}|
  \end{align*}
  for some $z_{0} \in \mathcal{T}$
  fixed. Applying~\autoref{lem:subgaussian-process-better-constant} to
  $\mathcal{T}\cup\{0\}$ gives
  \begin{align*}
    \E |X_{z_{0}}| \leq C \tilde K \|z_{0}\|_{2}. 
  \end{align*}
  Accordingly, by a standard relation between Gaussian width and complexity
  (\emph{cf.}~\autoref{prop:gaussian-complexity-finite-collection}),
  \begin{align*}
    \E \sup_{z \in \mathcal{T}} |X_{z}| \leq C\tilde K \w(\mathcal{T}) + \E|X_{z_{0}}| %
    \leq C\tilde K \gamma(\mathcal{T}). 
  \end{align*}
  An analogous set of steps to those above and in the proof
  for~\cite[Theorem~1]{chen2018stable} yields the high-probability bound.
\end{proof}


%% file: supplementary-numerics.tex
\section{Further Numerics Details}
\label{sec:s1-numerics}


In this section we describe the details of the numerical simulations appearing
in~\nameref{sec:numerics}. There, two generators were used to run a demixing
algorithm supporting the theory developed
in~\nameref{sec:gen-demix-random-mapping}. For each generator, a variational
autoencoder (VAE) was trained on a set of images and the VAE's decoder was used
as the generator. A pseudo-code meta-algorithm for the numerics is presented
in~\autoref{fig:meta-alg}. It may serve as a useful roadmap for the steps
presented below, though it is at best a terse version of the full description
provided in~\nameref{sec:numerics}. We begin by giving a brief overview of VAEs,
largely establishing where comprehensive details of our more-or-less standard
VAE implementation may be sought.

\subsection{Variational Autoencoders}
\label{sec:vari-auto}

We provide here some background relevant to our VAE implementation. We claim no
novelty of our VAE implementation, noting that a comprehensive overview of our
implementation is described in~\cite[Chapter~2]{kingma2019introduction}.

The goal of the variational inference problem is to selecting a model
$p_{\theta}(x)$ with parameters $\theta$ that approximately models an unknown
distribution $p^{*}(x)$ over data $x$. A common approach introduces a latent
variable model with the goal of learning parameters $\theta$ that capture a
joint distribution $p_{\theta}(x, z)$ with
$p_{\theta}(x) = \int p_{\theta}(x, z) \d z \approx p^{*}(x)$. Typically, it is
intractable to compute the marginal $p_{\theta}(x)$ and the posterior
distribution $p_{\theta}(z \mid x)$. Instead, the variational inference approach
introduces an \emph{inference model} $q_{\phi}(z\mid x)$ called the
\emph{encoder}. The parameters $\phi$ of the encoder are optimized such that
$q_{\phi}(z \mid x) \approx p_{\theta}(z \mid x)$. In the present work, the
encoder networks $\mathcal{E}(x; \phi)$ described above yield parameters
$(\mu, \log( \sigma^{2}))$ for the distribution $q_{\phi}(z \mid x)$, which is
assumed to be Gaussian. Samples from $q_{\phi}$, then are obtained by
``carefully'' sampling from a normal distribution. What this means will be
described below.

For VAEs, the objective function optimized is the evidence lower bound
(ELBO). The ELBO objective (equivalently, ELBO loss) is defined as
\begin{align}
  \label{eq:elbo-loss}
  \mathcal{L}_{\theta, \phi}(x) %
  := \E_{q_{\phi}(z \mid x)} \left[ \log p_{\theta}(x, z) - \log q_{\phi}(z\mid x) \right] %
  = \log p_{\theta}(x) - {\kldiv{q_{\phi}(z\mid x)}{p_{\theta}(z\mid x)}}.
\end{align}
The ELBO objective (equivalently, ELBO loss) can be efficiently optimized using
minibatch SGD and the \emph{reparametrization
  trick}~\cite[$\S\,2.4$]{kingma2019introduction}. Namely, in practice,
efficiently computing approximate gradients for minibatch SGD typically makes
use of the following steps. The reparametrization trick uses:
\begin{align}
  (\mu, \log \sigma^{2}) %
  & = \mathcal{E}(x; \phi)
    \nonumber
  \\
  \label{eq:vae-correspondence}
  z %
  & = \mu + \sigma \odot \varepsilon, \qquad \varepsilon \sim \mathcal{N}(0, I).
\end{align}
Above, $\odot$ denotes the element-wise product between its vector arguments. By
this computation, the encoder distribution assumes that the latent variables are
derived from a Gaussian distribution with diagonal covariance matrix. The
parameters of the Gaussian describing the conditional distribution are computed
by $\mathcal{E}(\cdot; \phi)$. Specifically:
\begin{align*}
  (\mu, \log \sigma^{2}) %
  & = \mathcal{E}(x; \phi)
  \\
  q_{\phi}(z \mid x) %
  & = \prod_{i} q_{\phi}(z_{i} \mid x) %
    = \prod_{i} \mathcal{N}(z_{i}; \mu_{i}, \sigma_{i}^{2})
\end{align*}
In the present work $\mathcal{E}(x; \phi)$ is an ``encoder'' neural network with
parameters $\phi$ and input image $x \in [0, 1]^{784}$, with architecture as
described in \autoref{fig:vae-arch-encoder}. A VAE acting on an input image $x$
is thus given by $A(x) := \mathcal{D}\circ (\mu + \sigma\odot\varepsilon)$ where
$\mu, \sigma, \varepsilon$ are given according to \eqref{eq:vae-correspondence},
$\mathcal{D}$ is the decoder network and $\mathcal{E}$ the encoder network.

The parameters $\mu, \sigma^{2} \in \reals^{128}$ are vectors whose elements
$\mu_{i}, \sigma_{i}^{2}$ are the mean and variance corresponding to each of the
latent variables $z_{i}$. Note that $z, \mu, \sigma^{2} \in \reals^{128}$, with
$128$ being the latent dimension, and that these quantities are referred to
in~\autoref{fig:vae-arch} by \inlinecode{z}, \inlinecode{mu} and
\inlinecode{log\_var}, respectively. In fact, \inlinecode{log\_var} represents
the log-variance $\log \sigma^{2}$, which is common to use due to improved
numerical stability during training.

In the training of each network, the loss function was the standard
implementation used for VAEs. A derivation of the loss function for VAEs in the
setting where the encoder is assumed to be Gaussian may be found
in~\cite[Chapter~2.5]{kingma2019introduction}. An alternative derivation is
given in~\cite{odaibo2019tutorial}. We provide an abbreviated description of the
loss function here for the sake of completeness. Let $F$ be a VAE, $x^{(i)}$ an
input image, $z^{(i)}$ its latent code, and $x_{r}^{(i)} = F(X)$ the image
reconstructed by the VAE. The loss function in our numerical implementation,
expressed for a single datum $x^{(i)}$, is given by:
\begin{align*}
  \mathcal{L}(x^{(i)}, F) %
  := \operatorname{BCE}(x_{r}^{(i)}, x^{(i)}) + \kldiv{q(z\mid x^{(i)})}{p(z)}.
\end{align*}
The former term is the binary cross-entropy between the image $x^{(i)}$ and its
reconstruction $x_{r}^{(i)}$. The latter term is the KL-divergence between the
encoder distribution $q(z\mid x^{(i)})$ and the latent prior $p(z)$. By
construction the latent prior $p(z)$ is assumed to tbe standard normal. This
loss function may be derived from the general one described
in~\eqref{eq:elbo-loss}. Implementation details for the binary cross-entropy are
available in the
\href{https://pytorch.org/docs/1.6.0/generated/torch.nn.BCELoss.html}{PyTorch
  documentation}; the KL-divergence in the present setting is equivalent to
$\kldiv{q_{\phi}(z_{j}\mid x^{(i)})}{p(z)} = \frac{1}{2}\left[\mu_{j}^{2} +
  \sigma_{j}^{2} - \log(\sigma_{j}^{2}) - 1\right]$. Theoretical background and
further detail, including on the training of VAEs, is given
in~\cite{doersch2016tutorial, kingma2013auto, kingma2019introduction}.

\begin{figure}[h]
  \centering
  \null\hfill
  \begin{subfigure}[t]{0.3\linewidth}
    \centering
    \includegraphics[width=\textwidth]{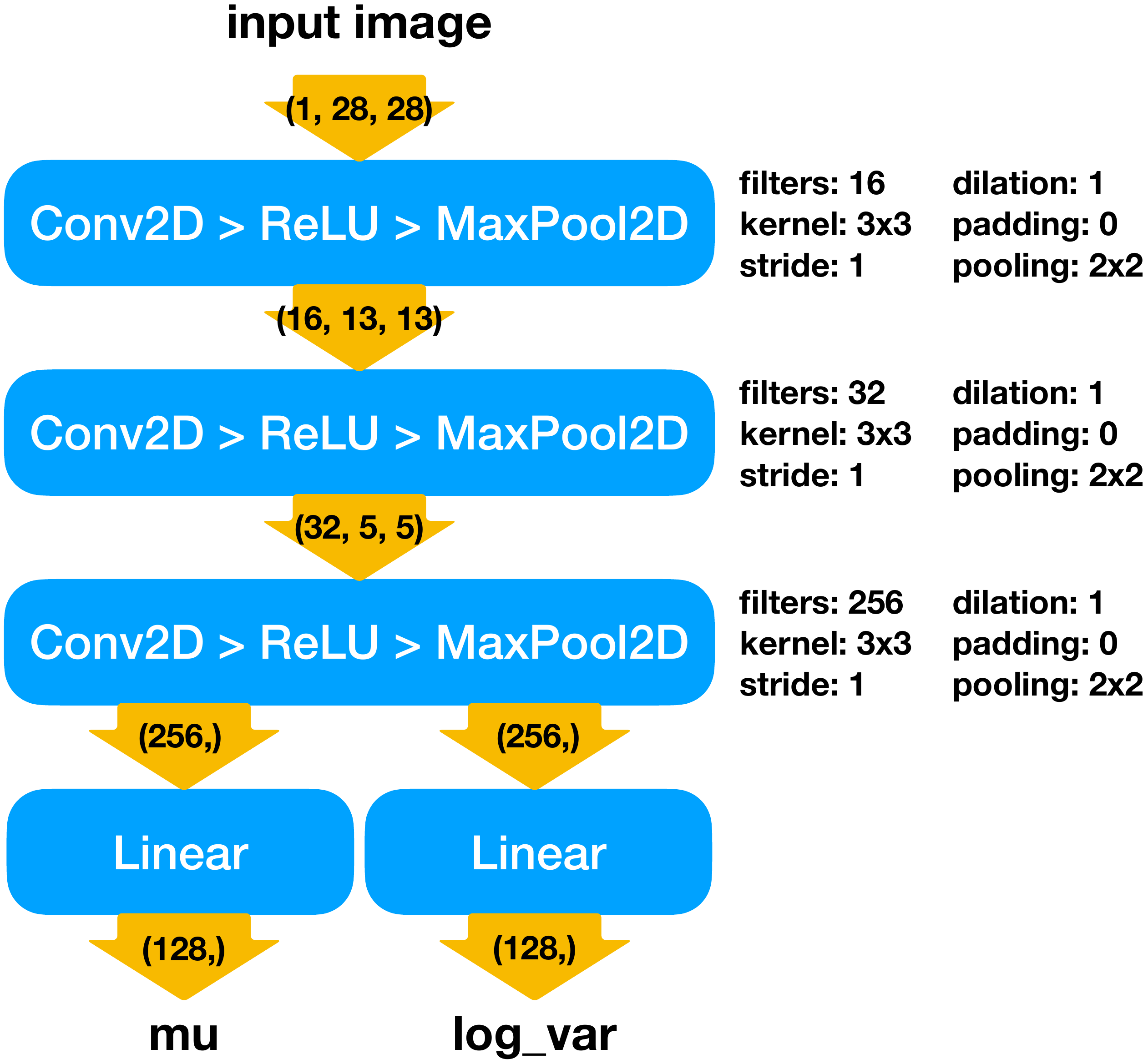}
    \caption{Encoder architecture}
    \label{fig:vae-arch-encoder}
  \end{subfigure}
  \hfill
  \begin{subfigure}[t]{0.3\linewidth}
    \centering
    \includegraphics[width=\textwidth]{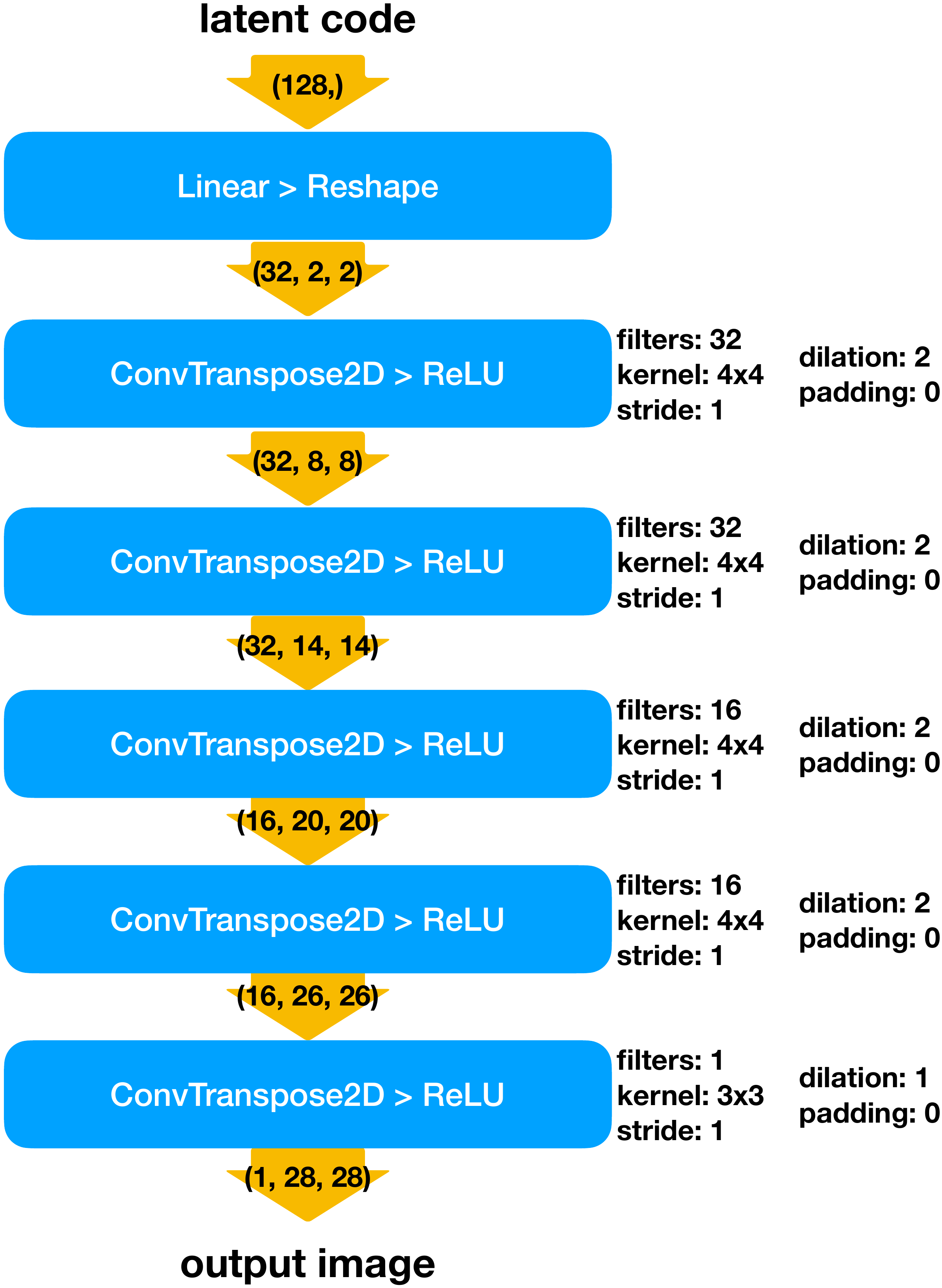}
    \caption{Decoder architecture}
    \label{fig:vae-arch-decoder}
  \end{subfigure}\hfill\null
  \caption[Autoencoder architecture]{Autoencoder: architecture for the encoder and decoder}
  \label{fig:vae-arch}
\end{figure}

\subsection{Creating the generators}
\label{sec:creating-generators}

The generators used in these numerical simulations were obtained as the decoders
of trained VAEs. The MNIST database~\cite{deng2012mnist} of images was used for
these simulations. The images of this database are grayscale $28\times 28$
images of handwritten digits. Each network only used images corresponding to a
single digit class. For this simulation, the digit class $1$ was used in
training the first VAE, for which there were 7877 images;\ $8$ for the second
VAE, for which there were 6825 images. For each network, the images were
randomly partitioned into training, validation and test sets (70\%, 15\%, 15\%,
respectively). No image augmentations (\eg random flips or random crops) were
used. For both networks, the training batch size was 32, and the number of
epochs used for training was 200. Early stopping was not used, as we found that
the validation loss continued to decrease for the duration of training.

Both networks were variational auto-encoders (VAEs) whose encoders used 2D
convolutional layers with max-pooling, and whose decoders used 2D transposed
convolutions for image generation. The architectures for the encoder and the
decoder are depicted in~\autoref{fig:vae-arch-encoder}
and~\ref{fig:vae-arch-decoder}, respectively. The latent dimension for both VAEs
was $128$. Both VAEs assumed that the conditional distribution on the latent
codes was Gaussian, as described in~\nameref{sec:vari-auto}. Both VAEs were
trained using a modification of mini-batch stochastic gradient descent
(SGD)~\cite{hardt2016train} with a learning rate of $10^{-5}$, momentum of $0.9$
and weight decay of $10^{-3}$, using the standard PyTorch implementation of
SGD~\cite{paszke2017automatic}. Results of the training are depicted
in~\autoref{fig:mnist-training-metrics}. In
particular,~\ref{fig:mnist-training-metrics-ones} for the VAE trained on $1$s,
and~\ref{fig:mnist-training-metrics-eights} for the VAE trained on $8$s.

Example outputs from noisy latent codes obtained from the validation data are
depicted in~\autoref{fig:mnist-vae-examples}. Specifically, this figure offers a
comparison of true images from the validation set and their auto-encoded
approximations. Validation images from the $1$s dataset can be found in the left
panel of~\autoref{fig:mnist-vae-val-true-recon-ones}, and their reconstructions
in the right panel. Validation images from the $8$s dataset can be found in the
left panel of~\autoref{fig:mnist-vae-val-true-recon-eights}; their
reconstructions in the right panel. The correspondence between the true image
and the auto-encoded approximation is given according to the image's row-column
position in the grid.

\begin{figure}[h]
  \centering
  \begin{subfigure}[h]{0.6\linewidth}
    \centering
    \includegraphics[width=\textwidth]{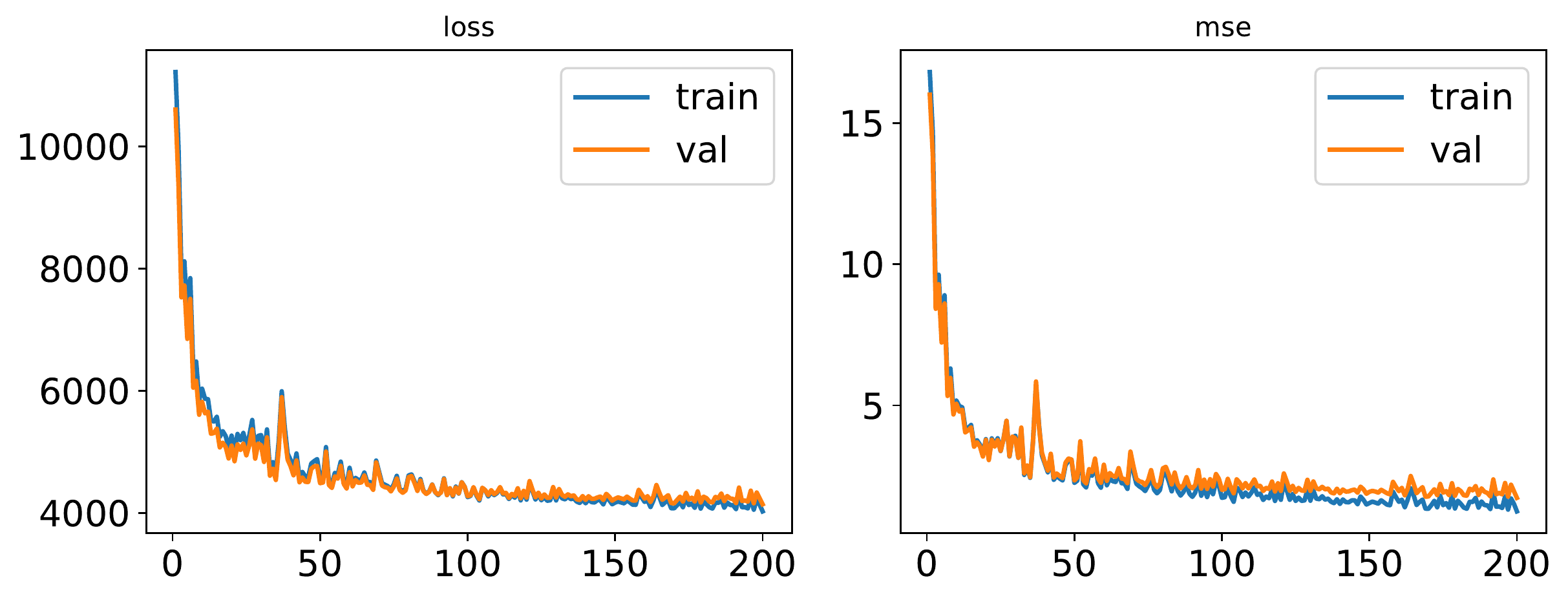}
    \caption{Loss \& MSE for VAE trained on $1$s.}
    \label{fig:mnist-training-metrics-ones}
  \end{subfigure}

  \centering
  \begin{subfigure}[h]{0.6\linewidth}
    \centering
    \includegraphics[width=\textwidth]{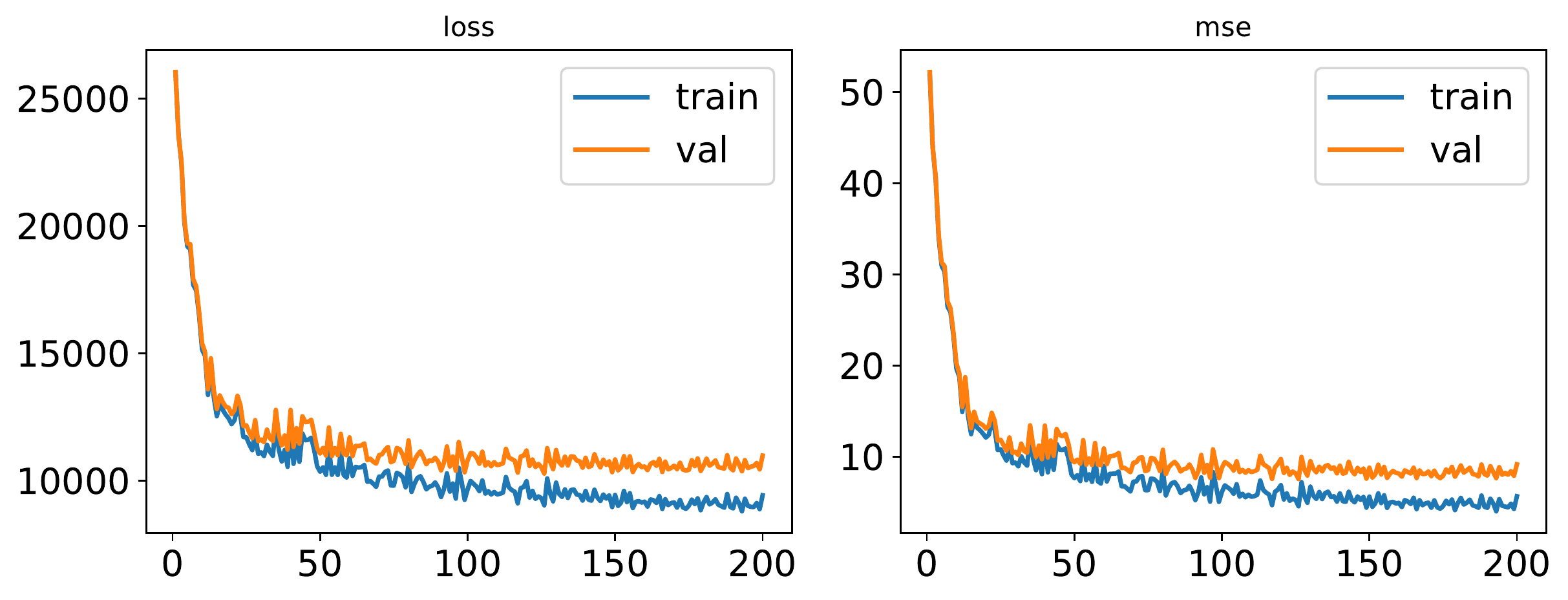}
    \caption{Loss \& MSE for VAE trained on $8$s.}
    \label{fig:mnist-training-metrics-eights}
  \end{subfigure}
  \caption{In each sub-figure, the left plot gives the average training and
    validation loss after each epoch; the right the training and validation MSE
    after each epoch.}
  \label{fig:mnist-training-metrics}
\end{figure}

\begin{figure}[h]
  \centering
  \begin{subfigure}[h]{0.45\linewidth}
    \centering
    \includegraphics[width=.45\textwidth]{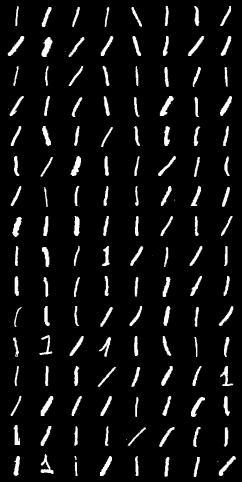}
    \quad
    \includegraphics[width=.45\textwidth]{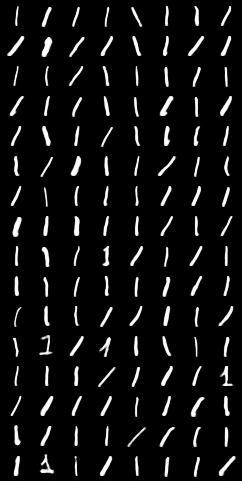}
    \caption{True and reconstructed images from the $1$s VAE}
    \label{fig:mnist-vae-val-true-recon-ones}
  \end{subfigure}
  \hfill
  \begin{subfigure}[h]{0.45\linewidth}
    \centering
    \includegraphics[width=.45\textwidth]{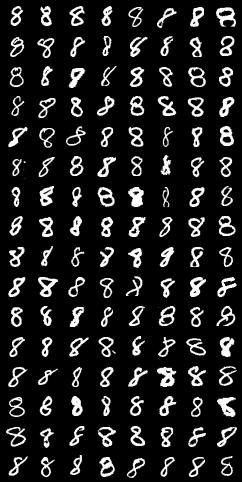}
    \quad
    \includegraphics[width=.45\textwidth]{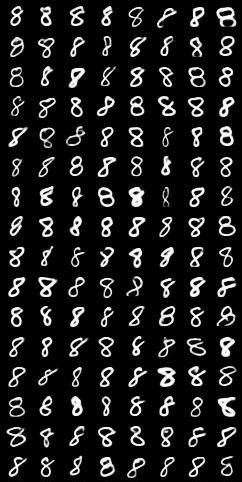}
    \caption{True and reconstructed images from the $8$s VAE}
    \label{fig:mnist-vae-val-true-recon-eights}
  \end{subfigure}
  \caption{Validation images and their reconstructions by each VAE. In each
    subfigure, true images from the validation data are shown in the left grid;
    their reconstruction by the decoder from codes that were partially corrupted
    are on the right.}
  \label{fig:mnist-vae-examples}
\end{figure}

\subsection{Generative demixing numerical implementation}
\label{sec:demixing-details}

Generally, the trained encoder network can be interpreted as a map from an image
to the encoder's \emph{latent code} for that image in the VAE's \emph{latent
  space}. The decoder can be interpreted as a mapping from the latent space back
into the relevant space of natural images. In this way, we are able to use the
decoder as a structural proxy corresponding to a low-dimensional representation
for an image, just as one might use a particular convex strucutural proxy like
the $\ell_{1}$ ball to encode a sparse signal in the classical compressed
sensing setting. For further intuition regarding this point we refer the reader
to~\nameref{sec:introduction} and~\cite{bora2017compressed,
  heckel2019deepdecoder, lempitsky2018deep}.

The initial mixture $b$ for the demixing problem was obtained as described
in~\eqref{eq:b-mnist}. As a minor technical point, we note that the values of
the resulting mixture $b$ were clipped to $[0,1]$, as were those of the
predicted mixture at each iteration of the recovery algorithm. We found that the
presence of clipping had no apparent effect on recovery performance, but aided
visual interpretation of the resulting images. The initial encodings
$w_{1}^{(0)}, w_{8}^{(0)} \in \reals^{128}$ used by each VAE decoder
$\mathcal{D}_{1}, \mathcal{D}_{8}$ were obtained for $b$ using the VAEs'
encoders $\mathcal{E}_{1}, \mathcal{E}_{8}$ and random noise:
\begin{align*}
  w_{c}^{(0)} = \mathcal{E}_{c}(b) + 0.1\cdot \varepsilon, %
  \qquad c \in \{1, 8\}, \varepsilon_{i} \iid \mathcal{N}(0, 1).
\end{align*}

To solve the demixing problem numerically, the MSE loss of the difference
between the predicted mixture and the true mixture was minimized using
Adam~\cite{kingma2014adam} with a learning rate of $10^{-2}$ for $1000$
iterations using PyTorch~\cite{paszke2017automatic}. Our proposed algorithm,
using PyTorch-themed pseudo-code, is presented
in~\autoref{fig:demix-alg-pytorch-pseudo-code}. For sake of clarity to the
reader, we have removed some minor code-specific details, and have used
suggestive variable names. The recovered image approximating $x_{1}$ had an MSE
of $1.59\cdot 10^{-3}$; that for $x_{8}$, $7.48 \cdot 10^{-4}$; that for the
mixture, $3.92\cdot 10^{-4}$. A graphical depiction of the results appear
in~\autoref{fig:mnist-vae-demixing-results} as
Figures~\ref{fig:mnist-vae-true0}--\ref{fig:mnist-vae-demix-loss-plot}.

\lstset{
  morekeywords={decode, zero_grad, view, backward, step},
}

\begin{figure*}[h]
  \centering
  \begin{center}
    \begin{minipage}[h]{0.7\linewidth}
\begin{lstlisting}
for i in range(num_iter):
    # compute de-mixed images from latent codes
    demixed_1 = vae_1.decode(w_1)
    demixed_8 = vae_8.decode(w_8)
    # clear gradients from previous iteration
    optimizer.zero_grad()
    # obtain predicted mixture
    mixture_pred = mixer(demixed_1, demixed_8)
    # compute MSE loss between predicted and true mixture
    loss = criterion(mixture_pred.view(1, -1), mixture.view(1, -1))
    # back-propagate to compute gradients
    loss.backward()
    # update w_1 and w_8
    optimizer.step()
\end{lstlisting}
    \end{minipage}
  \end{center}
  \caption{PyTorch pseudo-code for implementing the demixing algorithm.}
  \label{fig:demix-alg-pytorch-pseudo-code}
\end{figure*}

\lstset{
  morekeywords={GetEncoderArchitecture, GetDecoderArchitecture, SampleFrom, MinibatchSGD, VAELoss, MSELoss, AdamOptimizer, ComputeInitialEncodings, RunDemixing, function, NormalizedGaussianRandomMatrices, Mult},
  escapeinside={\#*}{*)}
}

\begin{figure*}[h]
  \centering
  \begin{center}
    \begin{minipage}[h]{0.7\linewidth}
\begin{lstlisting}
# Establish Architecture
Encoder #*$\leftarrow$*) GetEncoderArchitecture()
Decoder #*$\leftarrow$*) GetDecoderArchitecture()
AutoEncoder #*$\leftarrow$*) function(x, Encoder, Decoder):
    xReconstructed #*$\leftarrow$*) Decoder(SampleFrom(Encoder(x)))
    Return xReconstructed

# Train VAEs
TrainedDecoder0, TrainedEncoder0 #*$\leftarrow$*) MinibatchSGD(
    Decoder,
    Encoder,
    TrainingData0,
    ValidationData0,
    VAELoss,
    VAEOptimizationParameters,
)
TrainedDecoder1, TrainedEncoder1 #*$\leftarrow$*) MinibatchSGD(
    Decoder,
    Encoder,
    TrainingData1,
    ValidationData1,
    VAELoss,
    VAEOptimizationParameters,
)

# Demixing Algorithm
x0 #*$\leftarrow$*) SampleFrom(TestData0)
x1 #*$\leftarrow$*) SampleFrom(TestData1)
A #*$\leftarrow$*) SampleFrom(NormalizedGaussianRandomMatrices)
b #*$\leftarrow$*) Mult(A, x0) + x1
w_init0, w_init1 #*$\leftarrow$*) ComputeInitialEncodings()
DemixedImage0, DemixedImage1 #*$\leftarrow$*) RunDemixing(
    w_init0,
    w_init1,
    TrainedDecoder0,
    TrainedDecoder1,
    MSELoss,
    AdamOptimizer,
    DemixOptimizationParameters,
)
\end{lstlisting}
    \end{minipage}
  \end{center}
  \caption{A pseudo-code meta-algorithm serving as roadmap for the numerics of
    this work.}
  \label{fig:meta-alg}
\end{figure*}
